\documentclass[10pt,article,reqno]{amsart}
\usepackage{amssymb}
\usepackage{amsthm}
\usepackage{amsmath}
\usepackage{graphicx}
\usepackage{tikz}
\usepackage{tikz-3dplot}

\theoremstyle{plain}
\newtheorem{theorem}{Theorem}
\newtheorem{lem}{Lemma}[section]

\newcommand\im{\operatorname{Im}}
\newcommand\Tr{\operatorname{Tr}}
\newcommand\card{\#}
\newcommand\diam{\operatorname{diam}}

\def\irC{\mathcal{C}}

\def\R{\mathbb{R}}
\def\N{\mathbb{N}}
\def\C{\mathbb{C}}
\def\Q{\mathbb{Q}}

\def\K{\mathbb{K}}

\begin{document}

\title{Symmetries of projective spaces and spheres}
\author{Gy\"orgy P\'al Geh\'er}
\address{MTA-SZTE Analysis and Stochastics Research Group, Bolyai Institute, University of Szeged, H-6720 Szeged, Aradi v\'ertan\'uk tere 1., Hungary}
\address{MTA-DE "Lend\"ulet" Functional Analysis Research Group, Institute of Mathematics, University of Debrecen, H-4010 Debrecen, P.O. Box 12, Hungary}
\email{gehergy@math.u-szeged.hu or gehergyuri@gmail.com}
\urladdr{http://www.math.u-szeged.hu/$\sim$gehergy/}

\begin{abstract}
Let $H$ be either a complex inner product space of dimension at least two, or a real inner product space of dimension at least three. 
Let us fix an $\alpha\in \left(0,\tfrac{\pi}{2}\right)$. 
The purpose of this paper is to characterize all bijective transformations on the projective space $P(H)$ obtained from $H$ which preserves the angle $\alpha$ between lines in both directions.
(We emphasize that we do not assume anything about other angles).
For real inner product spaces and when $H=\C^2$ we do this for every $\alpha$, and when $H$ is a complex inner product space of dimension at least three we describe the structure of these transformations for $\alpha\leq\tfrac{\pi}{4}$. 
As an application, we give an Uhlhorn-type generalization of a famous theorem of Wigner which is considered to be a cornerstone of the mathematical foundations of quantum mechanics. 
Namely, we show that under the above assumptions, every bijective map on the set of pure states of a quantum mechanical system that preserves the transition probability $\cos^2\alpha$ in both directions is a Wigner symmetry (i.e. it automatically preserves all transition probability), except for the case when $H=\C^2$ and $\alpha = \tfrac{\pi}{4}$ where an additional possibility occurs.
We note that the classical theorem of Uhlhorn is the solution for the $\alpha = \tfrac{\pi}{2}$ case.
Usually in the literature, results which are connected to Wigner's theorem are discussed under the assumption of completeness of $H$, however, here we shall remove this unnecessary hypothesis.
Our main tools are a characterization of bijective maps on unit spheres of real inner product spaces which preserve an angle in both directions, and an extension of Uhlhorn's theorem for non-complete inner product spaces.
\end{abstract}

\maketitle


\section{Introduction}

Characterizing bijective isometries between normed spaces is a classical and important area of functional anaylsis.
For instance, the Mazur--Ulam theorem asserts that every bijective isometry between real normed spaces is an affine map (i.e. a composition of a linear transformation and a translation by a vector).
A consequence of this result is that if two real normed spaces are isomorphic as metric spaces, then they are also isomorphic as vector spaces.
Another classical result in this area is the Banach--Stone theorem, which characterizes bijective linear isometries between Banach spaces of continuous functions on compact Hausdorff spaces.
In particular, the existence of such a map implies that the underlying compact Hausdorff spaces are topologically equivalent.
Since the appearance of these remarkable theorems, the structures of bijective isometries have been described for several important normed spaces, and even today this is an active area of funcitonal analysis.
The books \cite{FJ1,FJ2} give an excellent introduction to this topic.
We also mention a few important papers which are connected to this area: \cite{lin1,lin2,lin3,lin4,lin5,lin6,lin7,lin8,lin9}.

Describing the structure of bijective isometries of non-linear spaces is also very important in functional analysis.
One of the most classical example of this is the famous theorem of Wigner \cite{Ba,C2,LoMe,M,R,ShAl,W}.
Let us state the classical version of Wigner's theorem in three different ways.
First, let $H$ be a complex (or real) Hilbert space, and $P_1(H)$ be the set of all rank-one orthogonal projections acting on $H$.
In quantum mechanics rank-one orthogonal projections represent pure states of the system.
The transition probability between $P,Q \in P_1(H)$ is defined by $\Tr PQ \in [0,1]$, where $\Tr$ denotes the trace functional.
Motivated by some problems in quantum physics, Wigner studied bijective maps $\phi\colon P_1(H)\to P_1(H)$ which preserve the transition probability, i.e.
$$
\Tr \phi(P)\phi(Q) = \Tr PQ \qquad (P,Q \in P_1(H)),
$$
and he showed that these mappings are always induced by unitary or antiunitary operators (orthogonal operators in the real case) acting on $H$.
These mappings are usually called Wigner symmetries.

Second, the following formula can be easily verified:
$$
\|P-Q\| = \sqrt{1-\Tr PQ} \qquad (P,Q \in P_1(H)),
$$
where $\|\cdot\|$ denotes the usual operator norm.
Therefore Wigner's theorem can be re-phrased in the following way: it characterizes bijective isometries of $P_1(H)$ with respect to the so-called gap metric, the distance obtained from the operator norm.
The gap metric was introduced and investigated by  Sz\H{o}kefalvi-Nagy and independently by  Krein and  Krasnoselski.
This notion has a wide range of applications from pure mathematics to engineering, in particular in perturbation theory of linear operators, perturbation analysis of invariant subspaces, optimization, robust control, multi-variable control, system identification and signal processing.

Third, let $P(H)$ denote the projective space obtained from $H$, i.e. the set of all one-dimensional subspaces of $H$.
We will call a one-dimensional subspace a line, for short.
Of course the map $P_1(H)\to P(H)$, $P\mapsto \im P$ gives a natural bijective transformation.
If $0\neq v\in H$ is a vector, then $[v]$ will stand for the line generated by it.
The angle between two lines $[u], [v]\in P(H)$ is defined by
$$
\measuredangle([u],[v]) := \arccos\frac{|\langle u,v\rangle|}{\|u\|\cdot\|v\|} \in \left[0,\tfrac{\pi}{2}\right].
$$
In case when $\measuredangle([u],[v]) = \tfrac{\pi}{2}$, then instead of this we shall usually write $[u]\perp [v]$, moreover, given any subset $\mathcal{L} \subset P(H)$ we will use the notation $\mathcal{L}^\perp = \{[w] \in P(H) \colon [w]\perp [u], \; \forall \; [u] \in \mathcal{L}\}$.
An easy calculation gives that if $P,Q\in P_1(H)$, $[u] = \im P$ and $[v] = \im Q$, then we have
$$
\Tr PQ = \left(\frac{|\langle u,v\rangle|}{\|u\|\cdot\|v\|}\right)^2 = \cos^2 \measuredangle([u],[v]).
$$
Therefore, Wigner's theorem can be viewed as a characterization of bijective transformations $\phi\colon P(H)\to P(H)$ which preserves the angle between lines, i.e.
$$
\measuredangle(\phi([u]),\phi([v])) = \measuredangle([u],[v]) \qquad ([u],[v]\in P(H))
$$
We point out that the above defined angle provides a metric on $P(H)$.
This is a folklore however, since we were not able to find it in the literature, we will provide a proof in Section \ref{proj_sec}.
Therefore Wigner's theorem can be also considered as a characterization of bijective isometries of $P(H)$ with respect to the metric $\measuredangle$.

There is a huge literature of generalizing Wigner's theorem in several directions.
Perhaps the first of these was provided by Uhlhorn (\cite{U}).
He managed to show that if the dimension of the Hilbert space is at least three, then any bijective map $\phi$ on $P(H)$ (or on $P_1(H)$, repsectively) which preserves orthogonality in both directions, i.e.
$$
\phi([u])\perp \phi([v]) \;\iff\; [u]\perp [v] \qquad ([u],[v]\in P(H))
$$
(or zero transition probability in both directions, resp.), is automatically a Wigner symmetry.
This is a serious improvement, since in Uhlhorn's theorem only the preservation of the quantum logical structure is assumed, while in Wigner's theorem its complete probabilistic structure is preserved.
We point out that it is a folklore that Uhlhorn's theorem can be stated also for inner product spaces.
However, as far as we know this general version has not been published anywhere, therefore we will provide this missing version in Section \ref{sphere_sec}.
For further generalizations of Wigner's theorem on projective spaces we mention the references \cite{BrMoSt,Br,C1,C2,M0,M2,RS1,RS2,S1,S2}. 
For generalizations where transformations on Grassmann spaces or on certain classes of idempotent operators are considered we mention \cite{BJM,GS,Gy,HH,M1,S3,S4}.

In this paper, following Uhlhorn's direction, we are interested in giving a natural generalization of Wigner's theorem.
Instead of preserving orthogonality of lines (or equivalently zero transition probability between pure states) in both directions, we will assume that a fixed angle $\alpha\in\left(0,\tfrac{\pi}{2}\right)$ (or equivalently a fixed transition probability $\cos^2\alpha$) is preserved in both directions.
We point out here that this problem has been partially answered in \cite{LPS} by Li, Plevnik and \v Semrl in the special case when $0 < \alpha \leq \tfrac{\pi}{4}$ and $H$ is a real Hilbert space with $4 < \dim H < \infty$.
Their method depends heavily on these rather restrictive assumptions.
Here, by developing a novel technique, we are able to solve this problem in a much more general setting.
Namely, in the real case we will answer the question completely, i.e. for every (not necessarily complete) inner product space with dimension at least three and for all angles $\alpha\in\left(0,\tfrac{\pi}{2}\right)$.
We will also solve the problem for every angle $\alpha\in\left(0,\tfrac{\pi}{2}\right)$ in the pure qubit case, i.e. when $H = \C^2$.
This could be quite a surprise since the general form of bijective transformations on $P(\C^2)$ preserving orthogonality in both directions is irregular due to the fact that every line has a unique orthogonal complement.
For complex inner product spaces of dimension at least three we will provide a characterization for angles $\alpha \in\left(0,\tfrac{\pi}{4}\right]$.

The outline of the paper is the following: in the next section we state four theorems which are the main results of this paper.
Then in Section \ref{sphere_sec} we will prove our result about transformations on spheres preserving an angle between unit vectors in both directions. 
This will enable us to prove our Uhlhorn-type generalization of Wigner's theorem in the real case which will be given in Section \ref{proj_sec}.
In the same section we will also prove our generalization in the complex case.

\section{Statements of the main results}

Let $H$ be an inner product space. 
If $x, y\in H$ are orthogonal, then we will write $x\perp y$.
If $H$ is a real space, then we will call a map $\phi\colon P(H) \to P(H)$ a Wigner symmetry, if there exists a bijective linear isometry $O\colon H\to H$ such that we have
\begin{equation}\label{WSR}
\phi([v]) = [Ov] \quad ([v] \in P(H)).
\end{equation}
In the complex case we call a transformation $\phi\colon P(H) \to P(H)$ a Wigner symmetry, if there exists a bijective linear or conjugatelinear isometry $U\colon H\to H$ such that we have
\begin{equation}\label{WSC}
\phi([v]) = [Uv] \quad ([v] \in P(H)).
\end{equation}
(If $H$ is a complex Hilbert space, then every bijective linear or conjugatelinear isometry is a unitary or an antiunitary operator.)
Clearly, Wigner symmetries preserve the angle between lines.
Conversely, R\"atz's version of Wigner's theorem in inner product spaces states that every bijective map $\phi\colon P(H) \to P(H)$ preserving the angle between lines is a Wigner symmetry (see \cite{R}).

Now, we state our generalization of Wigner's theorem in real inner product spaces.

\begin{theorem}\label{Umain-real}
Let $H$ be a real inner product space with $3\leq \dim H$. 
Suppose that $\alpha\in\left(0,\tfrac{\pi}{2}\right)$ and that $\phi\colon P(H) \to P(H)$ is a bijective map which preserves the angle $\alpha$ between lines in both directions, i.e. $\phi$ satisfies
$$
\measuredangle([u],[v]) = \alpha \;\iff\; \measuredangle(\phi([u]),\phi([v])) = \alpha \quad ([u],[v] \in P(H)).
$$
Then $\phi$ is a Wigner symmetry.
\end{theorem}

The proof of Theorem \ref{Umain-real} will require the study of bijective transformations of the unit sphere $S_H = \{h\in \|h\| = 1\}$ of a real inner product space $H$ which preserve an angle in both directions.
The angle between $x,y\in S_H$ is defined by 
$$
\sphericalangle(x,y) = \arccos\langle x,y\rangle \in [0,\pi].
$$
One could define the angle between any two non-zero vectors, but we will not need it here.
We note that writing $\sphericalangle(x,y) = \tfrac{\pi}{2}$ is the same as $x\perp y$.
The study of such maps was initiated by U. Everling in \cite{E}, where he proved a result for finite and at least three dimensional Hilbert spaces and for angles less than or equal to $\tfrac{\pi}{2}$.
Apparently, he was motivated by the surprising Beckman--Quarles theorem \cite{BQ} which states that every map on a finite dimensional real Hilbert space of dimension at least two preserving the unit distance is automatically an isometry.
However, we mention that in the statment of Everling's theorem (and also in the Beckman--Quarles theorem) the preservation property is only assumed in one direction, and the bijectivity property is also relaxed.
Though, usually in infinite dimensions, we cannot expect a regular structure for such general maps (see \cite[Section 1]{BQ} for an example).

Now, we state our improvement of Everling's theorem.

\begin{theorem}\label{Emain}
Let $H$ be a real inner product space with $3\leq \dim H$, and $0 < \alpha < \pi$.
Assume that $\psi \colon S_H \to S_H$ is a bijective transformation which preserves the angle $\alpha$ in both directions, i.e. we have
$$
\sphericalangle(x,y) = \alpha \;\iff\; \sphericalangle(\psi(x),\psi(y)) = \alpha \quad (x,y\in S_H).
$$
Then there exists a bijective linear isometry $R\colon H\to H$ such that 
\begin{itemize}
\item[(i)] if $\alpha \neq \tfrac{\pi}{2}$, then we have
\begin{equation}\label{sphere_nice-eq}
\psi(x) = Rx \quad (x\in S_H),
\end{equation}
\item[(ii)] if $\alpha = \tfrac{\pi}{2}$, then we have
$$
\psi(x) \in \{-Rx, Rx\} \quad (x\in S_H).
$$
\end{itemize}
\end{theorem}

We mention that in Theorems \ref{Umain-real} and \ref{Emain} one could provide characterizations also in the case when $\dim H = 2$.
However, the structures of such maps in this low dimensional case are not regular and therefore they would be useless for our purposes.

Next, let us consider a complex inner product space $H$ and an angle $\alpha\in\left(0,\tfrac{\pi}{2}\right)$.
Since Uhlhorn's theorem does not hold for two-dimensional spaces, it could be quite surprising that if $H = \C^2$, then every bijection preserving the angle $\alpha$ in both directions is actually a Wigner symmetry, except when $\alpha = \tfrac{\pi}{4}$.
In other words, bijections defined on the set of pure qubit states which preserve the transition probability $\cos^2\alpha$ have a regular structure.

For any $[u]\in P(\C^2)$ the symbol $[u]^\perp$ denotes the unique line which is orthogonal to $[u]$.
We state our theorem about pure qubit states.

\begin{theorem}\label{Umain-complex-2d}
Suppose that $\alpha\in\left(0,\tfrac{\pi}{2}\right)$ and that $\phi\colon P(\C^2) \to P(\C^2)$ is a bijective map such that
$$
\measuredangle([u],[v]) = \alpha \;\iff\; \measuredangle(\phi([u]),\phi([v])) = \alpha \quad ([u],[v] \in P(\C^2))
$$
holds.
Then 
\begin{itemize}
\item[(i)] in the $\alpha \neq \tfrac{\pi}{4}$ case $\phi$ is a Wigner symmetry,
\item[(ii)] in the $\alpha = \tfrac{\pi}{4}$ case there exists a unitary or an antiunitary operator $U\colon \C^2\to \C^2$ such that
$$
\phi([v]) \in \left\{[Uv],[Uv]^\perp\right\} \quad ([v] \in P(\C^2)).
$$
\end{itemize}
\end{theorem}

Finally, our last result is a generalization of Wigner's theorem for complex inner product spaces, which we state below.

\begin{theorem}\label{Umain-complex}
Suppose that $H$ is a complex inner product space, $3\leq \dim H$ and $\alpha\in\left(0,\tfrac{\pi}{4}\right]$.
We assume that $\phi\colon P(H) \to P(H)$ is a bijective map such that
$$
\measuredangle([u],[v]) = \alpha \;\iff\; \measuredangle(\phi([u]),\phi([v])) = \alpha \quad ([u],[v] \in P(H))
$$
holds.
Then $\phi$ is a Wigner symmetry.
\end{theorem}


\section{Symmetries of speheres}\label{sphere_sec}

The aim of this section is to verify Theorem \ref{Emain}, but before that we need to prove several lemmas.
We note that our method partially follows that of Everling's, however at several points it differs essentially.
Let $H$ be a real inner product space of dimension at least three.
It is known that the angle $\sphericalangle$ defines a metric on $S_H$.
A map $\psi\colon S_H \to S_H$ is said to be a $\sphericalangle$-isometry if we have 
$$
\sphericalangle(x,y) = \sphericalangle(\psi(x),\psi(y)) \quad (x,y\in S_H),
$$
or equivalently
$$
\langle x,y\rangle = \langle \psi(x),\psi(y)\rangle \quad (x,y\in S_H).
$$

We begin with the following characterization of bijective $\sphericalangle$-isometries.

\begin{lem}
Let $H$ be a real inner product space with $\dim H \geq 3$.
If $\psi \colon S_H \to S_H$ is a bijective $\sphericalangle$-isometry, then \eqref{sphere_nice-eq} is satisfied with some bijective linear isometry $R\colon H\to H$.
\end{lem}

\begin{proof}
We define the following bijective map:
$$
R\colon H\to H, \quad Rx = \left\{ \begin{matrix}
\|x\|\cdot \psi\left(\tfrac{1}{\|x\|}x\right) & \text{if } x\neq 0, \\
0 & \text{if } x = 0.
\end{matrix} \right.
$$
We would like to show that $R$ is a linear isometry. 
We obviously have $\langle Rx, Ry \rangle = \langle x, y \rangle$ $(x,y \in H)$. 
Therefore we obtain $\|Rx - Ry\|^2 = \langle Rx, Rx \rangle + \langle Ry, Ry \rangle - 2 \langle Rx, Ry \rangle = \langle x, x \rangle + \langle y, y \rangle - 2 \langle x, y \rangle = \|x - y\|^2$, hence $R$ is an isometry.
By the famous Mazur--Ulam theorem we immediately obtain the linearity of $R$.
Finally, it is easy to see that \eqref{sphere_nice-eq} is satisfied with $R$.
\end{proof}

For some vectors $z_1,\dots z_n\in H$ the symbol $[z_1,\dots z_n]$ will stand for the generated subspace in $H$.
If $\mathcal{A} \subseteq S_H$, then we will use the notation 
$$
\mathcal{A}^{(\alpha)} := \{y \in S_H \colon \sphericalangle(x,y) = \alpha \;\, \forall\, x\in \mathcal{A}\}. 
$$
In case when $\mathcal{A} = \{x\}$, we will simply write $x^{(\alpha)}$ instead of $\{x\}^{(\alpha)}$.
Next, we investigate the set $x^{(\alpha)} \cap y^{(\beta)}$ which will be later utilized in order to obtain other angles which are preserved by $\psi$.

\begin{lem}\label{sphere-0-1-element_lem}
Let $H$ be a real inner product space with $\dim H \geq 3$.
Let $x,y\in S_H$, $0 < \alpha < \beta < \pi$, and assume that $\gamma := \sphericalangle(x,y) \in (0,\pi]$. 
Then the following equivalences are satisfied:
\begin{equation}\label{a}
x^{(\alpha)} \cap y^{(\beta)} = \emptyset \;\iff\; 
\left\{ 
\begin{matrix}
\gamma > \alpha+\beta \text{ or } \gamma < \beta-\alpha, & \text{if } \alpha+\beta < \pi,\\
\gamma < \beta-\alpha, & \text{if } \alpha+\beta = \pi,\\
\gamma > 2\pi - \alpha-\beta \text{ or } \gamma < \beta-\alpha, & \text{if } \alpha+\beta > \pi,
\end{matrix} 
\right.
\end{equation}
and
\begin{equation}\label{b}
\card(x^{(\alpha)} \cap y^{(\beta)}) = 1 \;\iff\; 
\left\{ 
\begin{matrix}
\gamma \in \{\alpha+\beta, \beta-\alpha\}, & \text{if } \alpha+\beta < \pi,\\
\gamma = \beta-\alpha, & \text{if } \alpha+\beta = \pi,\\
\gamma \in \{2\pi - \alpha-\beta, \beta-\alpha\}, & \text{if } \alpha+\beta > \pi.
\end{matrix} 
\right.
\end{equation}
Moreover, if $\alpha \neq \tfrac{\pi}{2}$, then we have
$$
x^{(\alpha)} \cap y^{(\alpha)} = \emptyset \;\iff\; 
\left\{ 
\begin{matrix}
\gamma > 2\alpha , & \text{if } \alpha < \tfrac{\pi}{2},\\
\gamma > 2\pi - 2\alpha, & \text{if } \alpha > \tfrac{\pi}{2},
\end{matrix} 
\right.
$$
and
$$
\card(x^{(\alpha)} \cap y^{(\alpha)}) = 1 \;\iff\; 
\left\{ 
\begin{matrix}
\gamma = 2\alpha , & \text{if } \alpha < \tfrac{\pi}{2},\\
\gamma = 2\pi - 2\alpha, & \text{if } \alpha > \tfrac{\pi}{2}.
\end{matrix} 
\right.
$$
\end{lem}

\begin{proof}
We will consider only the case when $\alpha$ and $\beta$ are different, since the other one can be verified in a very similar way.
Let us observe that the quantity $\card(x^{(\alpha)} \cap y^{(\beta)} \cap K)$ is the same for every three-dimensional subspace $K$ of $H$ which contains $x$ and $y$.
Thus we have $x^{(\alpha)} \cap y^{(\beta)} = \emptyset$ if and only if we have $x^{(\alpha)} \cap y^{(\beta)} \cap K = \emptyset$ for some (or equivalently, for all) $K$.
Moreover, if $\card(x^{(\alpha)} \cap y^{(\beta)} \cap K) = 1$ with some $K$, and this unique element is denoted by $u$, then the symmetry of the sphere implies $u \in [x,y]$.
Therefore we conclude that we have $\card(x^{(\alpha)} \cap y^{(\beta)}) = 1$ if and only if we have $\card(x^{(\alpha)} \cap y^{(\beta)} \cap K) = 1$ for some (or equivalently, for all) $K$.
In particular this cannot happen if $\gamma = \pi$, since in this case we have $[x,y] = [x]$ and thus $u\in\{x,y\}$, which is a contradiction.
Therefore we only have to prove our statement for the case when $\dim H = 3$, so from now on we will assume this.

The first spherical law of cosines implies that $x^{(\alpha)} \cap y^{(\beta)} \neq \emptyset$ if and only if there exists a $C\in [0,\pi]$ such that
$$
\cos\gamma = \cos\alpha\cos\beta + \sin\alpha\sin\beta\cos C.
$$
This is equivalent to $\cos\gamma \in [\cos(\alpha+\beta),\cos(\beta-\alpha)] = [\cos(2\pi-\alpha-\beta),\cos(\beta-\alpha)]$, which implies \eqref{a}.

Clearly, $\card(x^{(\alpha)} \cap y^{(\beta)}) = 1$ implies $C\in\{0,\pi\}$.
Therefore one easily concludes that $\card(x^{(\alpha)} \cap y^{(\beta)}) = 1$ holds exactly when one of the following two conditions is satisfied:
\begin{itemize}
\item $C = 0$ and $\gamma = \beta-\alpha$, or
\item $C = \pi$, $\alpha+\beta \neq \pi$ and $\gamma = \left\{\begin{matrix} \alpha+\beta, & \text{if } \alpha+\beta<\pi, \\ 2\pi-\alpha-\beta, & \text{if } \alpha+\beta>\pi. \end{matrix}\right.$
\end{itemize}
The equivalences stated in the lemma follows from this.
\end{proof}

It is a standard and natural method in the theory of preserver problems that one characterizes a relation in terms of the property which is preserved by our bijective map in both directions, and then concludes that this map also preserves this relation in both directions.
Here this method will be utilized several times.

We proceed with the verifications of the following two lemmas.

\begin{lem}\label{multiple-sphericalangle-eq-leq_lem}
Let $0<\alpha<\tfrac{\pi}{2}$.
Under the hypotheses of Theorem \ref{Emain} we have the following two properties:
\begin{equation}\label{multiple-angle-eq_eq}
\sphericalangle(x,y) = j \alpha \;\iff\; \sphericalangle(\psi(x),\psi(y)) = j \alpha \quad (x,y\in S_H, j\in\N,\, 2\leq j < \tfrac{\pi}{\alpha}).
\end{equation}
and
\begin{equation}\label{multiple-angle-leq_eq}
\sphericalangle(x,y) \leq j \alpha \;\iff\; \sphericalangle(\psi(x),\psi(y)) \leq j \alpha \quad (x,y\in S_H, j\in\N,\, 2\leq j < \tfrac{\pi}{\alpha})
\end{equation}
\end{lem}

\begin{proof}
Assume that $j\in\N$ and $2 \leq j < \tfrac{\pi}{\alpha}$.
It is straightforward that we have $\sphericalangle(x,y) = j \alpha$ if and only if there is a unique $(j-1)$-element sequence $\{x_1,\dots x_{j-1}\}$ such that $\sphericalangle(x_l, x_{l+1}) = \sphericalangle(x, x_1) = \sphericalangle(x_{j-1}, y) = \alpha$ $(l=1,\dots j-2)$.
Since $\psi$ and $\psi^{-1}$ are bijective transformations preserving the angle $\alpha$ in both directions, we conclude that the latter condition is equivalent to the existence of a unique $(j-1)$-element sequence $\{y_1,\dots y_{j-1}\}$ such that $\sphericalangle(y_l, y_{l+1}) = \sphericalangle(\psi(x), y_1) = \sphericalangle(y_{j-1}, \psi(y)) = \alpha$ $(l=1,\dots j-2)$.
But this holds exactly when $\sphericalangle(\psi(x), \psi(y)) = \alpha$.
Therefore we obtain \eqref{multiple-angle-eq_eq}.

The proof of \eqref{multiple-angle-leq_eq} is very similar, we only have to observe that $\sphericalangle(x,y) \leq j \alpha$ is satisfied if and only if there exists a $(j-1)$-element sequence $\{x_1,\dots x_{j-1}\}$ such that $\sphericalangle(x_l, x_{l+1}) = \sphericalangle(x, x_1) = \sphericalangle(x_{j-1}, y) = \alpha$ $(l=1,\dots j-2)$.
\end{proof}

\begin{lem}\label{double-angle-leq_lem}
Let $H$ be a real inner product space with $\dim H \geq 3$, and $\psi \colon S_H \to S_H$ be a bijective map.
Assume that $0 < \alpha < \tfrac{\pi}{2}$ and that $\psi$ satisfies the following condition:
$$
\sphericalangle(x,y) \leq \alpha \;\iff\; \sphericalangle(\psi(x),\psi(y)) \leq \alpha \quad (x,y\in S_H).
$$
Then $\psi$ satisfies 
\begin{equation}\label{double-angle-eq_eq}
\sphericalangle(x,y) = 2 \alpha \;\iff\; \sphericalangle(\psi(x),\psi(y)) = 2 \alpha \quad (x,y\in S_H).
\end{equation}
\end{lem}

\begin{proof}
We only have to observe that $\sphericalangle(x,y) = 2 \alpha$ holds if and only if there is a unique $u\in S_H$ with $\sphericalangle(x, u) \leq \alpha$ and $\sphericalangle(u, y) \leq \alpha$.
\end{proof}

In the previous lemma we could have shown the property \eqref{multiple-angle-eq_eq} as well however, we will not need it in the sequel.
Next, we show that the $\sphericalangle$-isometriness of $\psi$ is a consequence of a much milder assumption.

\begin{lem}\label{small-sphericalangles_lem}
Let $H$ be a real inner product space with $\dim H \geq 3$, and $\psi \colon S_H \to S_H$ be a bijective map.
Suppose that there exists a decreasing sequence $\{\alpha_n\}_{n=1}^\infty \subseteq (0,\pi)$ with $\lim_{n\to\infty} \alpha_n = 0$ such that $\psi$ preserves all of these angles in both directions. 
Then $\psi$ is a $\sphericalangle$-isometry.
\end{lem}

\begin{proof} 
From \eqref{multiple-angle-leq_eq} of Lemma \ref{multiple-sphericalangle-eq-leq_lem} we conclude that
\begin{equation}\label{angle_interval_eq}
(j-1) \alpha_{n} < \sphericalangle(x,y) \leq j \alpha_{n} \;\iff\; (j-1) \alpha_{n} < \sphericalangle(\psi(x),\psi(y)) \leq j \alpha_{n}
\end{equation}
for every $x,y\in S_H$ and $j,n\in\N$, $3 \leq j < \tfrac{\pi}{\alpha_{n}}$.
Since $\alpha_n$ can be arbitrarily small, we easily obtain that $\psi$ preserves every angle which is less than $\pi$ in both directions.
Hence we also have
$$
\sphericalangle(x,y)<\pi \;\iff\; \sphericalangle(\psi(x),\psi(y))<\pi,
$$
and thus a negation of this equivalence gives that the angle $\pi$ is also preserved in both directions.
This completes the proof.
\end{proof}

Using Lemma \ref{sphere-0-1-element_lem} we can provide other angles which are preserved by $\psi$.

\begin{lem}\label{sphere-diff-sum_lem}
Let $H$ be a real inner product space with $\dim H \geq 3$, $0 < \alpha < \beta < \pi$, $\alpha+\beta \leq \pi$, and $\psi \colon S_H \to S_H$ be a bijective map.
Suppose that $\psi$ preserves the angles $\alpha$ and $\beta$ in both directions. 
Then 
\begin{itemize}
\item[(i)] $\psi$ preserves the angles $\beta-\alpha$ and $\alpha+\beta$ in both directions, if $\alpha + \beta < \pi$,
\item[(ii)] $\psi$ preserves the angle $\beta-\alpha$ in both directions, if $\alpha + \beta = \pi$.
\end{itemize}
\end{lem}

\begin{proof} (i): Let us observe that the properties of $\psi$ implies $\psi\left(x^{(\alpha)} \cap y^{(\beta)}\right) = \psi(x)^{(\alpha)} \cap \psi(y)^{(\beta)}$.
Therefore, by Lemma \ref{sphere-0-1-element_lem} we have 
$$
\sphericalangle(x,y) \in \{\beta-\alpha, \alpha+\beta\} \;\iff\; \card(x^{(\alpha)} \cap y^{(\beta)}) = 1
$$
$$
\;\iff\; \card(\psi(x)^{(\alpha)} \cap \psi(y)^{(\beta)}) = 1 \;\iff\; \sphericalangle(\psi(x),\psi(y)) \in \{\beta-\alpha, \alpha+\beta\}
$$
We observe that $\alpha<\tfrac{\pi}{2}$.
Clearly, there exists a $j\in\N$, $2\leq j < \tfrac{\pi}{\alpha}$ such that $\beta-\alpha < j \alpha < \alpha+\beta$.
Since 
$$
\sphericalangle(x,y) = \beta-\alpha \;\iff\; \sphericalangle(x,y) \in \{\beta-\alpha, \alpha+\beta\} \;\text{and}\; \sphericalangle(x,y)\leq j\alpha,
$$
the above equivalence and \eqref{multiple-angle-leq_eq} of Lemma \ref{multiple-sphericalangle-eq-leq_lem} imply that $\psi$ preserves the angle $\beta-\alpha$, and thus also $\alpha+\beta$ in both directions.

(ii): This is immediate from Lemma \ref{sphere-0-1-element_lem}.
\end{proof}

For a given set $\mathcal{A}\subseteq S_H$ we will use the notation $\diam_\sphericalangle(\mathcal{A}) = \sup\{\sphericalangle(x,y)\colon x,y\in\mathcal{A}\}$. 
We proceed with the following crucial statement.

\begin{lem}\label{main-sphere_lem}
Let us assume that the hypotheses of Theorem \ref{Emain} are satisfied for $\psi$ and that we have $0 < \alpha < \tfrac{\pi}{2}$.
Then $\arccos\left( \frac{4\cos^2\alpha}{\cos\alpha + 1} - 1 \right) \in (\alpha,2\alpha)$ and $\psi$ preserves this angle in both directions.
\end{lem}

\begin{proof}
Let $x,y\in S_H$ with $\gamma := \sphericalangle(x,y) \in (0,2\alpha) \subset \left(0,\pi\right)$.
By the bjectivity of $\psi$ and \eqref{multiple-angle-eq_eq} and \eqref{multiple-angle-leq_eq} of Lemma \ref{multiple-sphericalangle-eq-leq_lem}, this is equivalent to $\sphericalangle(\psi(x),\psi(y)) \in (0,2\alpha)$.
Let $u \in x^{(\alpha)} \cap y^{(\alpha)}$.
Then there exists an orthonormal system $\{e_1,e_2\} \subset H$ such that $x = \cos\tfrac{\gamma}{2} \cdot e_1 + \sin\tfrac{\gamma}{2} \cdot e_2$, $y = \cos\tfrac{\gamma}{2} \cdot e_1 - \sin\tfrac{\gamma}{2} \cdot e_2$.
Moreover, we have $u = \cos\delta \cdot e_1 + \sin\delta\cos\varepsilon \cdot e_2 + \sin\delta\sin\varepsilon \cdot e_3$ with some $e_3\in S_H$, $e_3 \perp e_j$ $(j=1,2)$, and $\delta, \varepsilon \in [0,\pi]$ which satisfy the following two equations:
$$
\langle u, x \rangle = \cos\tfrac{\gamma}{2}\cos\delta + \sin\tfrac{\gamma}{2}\sin\delta\cos\varepsilon = \cos\alpha
$$
and
$$
\langle u, y \rangle = \cos\tfrac{\gamma}{2}\cos\delta - \sin\tfrac{\gamma}{2}\sin\delta\cos\varepsilon = \cos\alpha.
$$
Thus we conclude that we have either $\delta\in\{0,\pi\}$, or $\varepsilon = \tfrac{\pi}{2}$.
It is easy to see that the first possibility cannot happen.
Indeed, in that case we would get $u \in \{e_1,-e_1\}$, hence $\alpha = \sphericalangle(u,x) = \sphericalangle(u,y) \in \left\{\tfrac{\gamma}{2},\pi-\tfrac{\gamma}{2}\right\}$ would follow, which would be a contradiction.
In the second case we have $\cos\delta = \frac{\cos\alpha}{\cos(\gamma/2)} \in (0,1)$.
Therefore we get the following equation
$$
x^{(\alpha)} \cap y^{(\alpha)} = \left\{ \frac{\cos\alpha}{\cos\tfrac{\gamma}{2}} \cdot e_1 + \sqrt{1-\frac{\cos^2\alpha}{\cos^2\tfrac{\gamma}{2}}} \cdot e_3 \colon \|e_3\| = 1,\, e_3 \perp e_1, \, e_3 \perp e_2 \right\} \neq \emptyset.
$$
It is straightforward that 
$$
h(\gamma) := \diam_\sphericalangle(x^{(\alpha)} \cap y^{(\alpha)}) = \arccos\left( 2 \frac{\cos^2\alpha}{\cos^2\tfrac{\gamma}{2}} - 1 \right).
$$ 
One easily calculates the unique solution $\gamma_0\in(0,2\alpha)$ of the equation $h(\gamma_0) = \alpha$ which is
$$ 
\gamma_0 = \arccos\left( \frac{4\cos^2\alpha}{\cos\alpha + 1} - 1 \right).
$$
Clearly, the function $h\colon (0,2\alpha) \to (0,2\alpha)$ is a strictly decreasing bijection.
Since 
$$
h(\alpha) > \alpha 
\;\iff\; \frac{2\cos^2\alpha}{\cos^2\tfrac{\alpha}{2}} - 1 < \cos\alpha 
\;\iff\; \cos \alpha \frac{4\cos^2\tfrac{\alpha}{2} - 2}{\cos^2\tfrac{\alpha}{2}} - 1 < \cos\alpha
$$
$$
\iff\; \cos \alpha \frac{3\cos^2\tfrac{\alpha}{2} - 2}{\cos^2\tfrac{\alpha}{2}} < 1
\;\iff\; \left( 3-\frac{2}{\cos^2\tfrac{\alpha}{2}} \right)\cos \alpha < 1
$$
and this latter inequality is satisfied, we obtain $h(\alpha) > \alpha$.
Thus the monotonicity of $h$ implies $\alpha < \gamma_0 < 2\alpha$.
Therefore we have just proven that there is a unique angle $\gamma_0\in(0,2\alpha)$ such that $\diam_\sphericalangle(x^{(\alpha)} \cap y^{(\alpha)}) = \alpha$ holds if and only if $\sphericalangle(x,y) = \gamma_0$, furthermore we have $\gamma_0\in(\alpha,2\alpha)$.

Now, an elementary observation shows that for every $u\in x^{(\alpha)} \cap y^{(\alpha)}$ and $\vartheta > 0$ we have $\card(x^{(\alpha)} \cap y^{(\alpha)} \cap u^{(\vartheta)}) = 1$ if and only if $\vartheta = h(\gamma)$.
In particular, we have $\card(x^{(\alpha)} \cap y^{(\alpha)} \cap u^{(\alpha)}) = 1$ for every $u\in x^{(\alpha)} \cap y^{(\alpha)}$ exactly when $\gamma = \gamma_0$, which gives a characterization of the angle $\gamma_0$ in terms of $\alpha$.
Now, similarly as in Lemma \ref{sphere-diff-sum_lem}, we obtain that $\psi$ preserves the angle $\gamma_0$ in both directions.
Namely, we have
$$
\sphericalangle(x,y) = \gamma_0 \;\iff\; \card(x^{(\alpha)} \cap y^{(\alpha)} \cap u^{(\alpha)}) = 1 \; (\forall \, u\in x^{(\alpha)} \cap y^{(\alpha)})
$$
$$
\iff\; \card(\psi(x)^{(\alpha)} \cap \psi(y)^{(\alpha)} \cap \psi(u)^{(\alpha)}) = 1 \; (\forall \, \psi(u)\in \psi(x)^{(\alpha)} \cap \psi(y)^{(\alpha)})
$$
$$
\iff\; \card(\psi(x)^{(\alpha)} \cap \psi(y)^{(\alpha)} \cap v^{(\alpha)}) = 1 \; (\forall \, v\in \psi(x)^{(\alpha)} \cap \psi(y)^{(\alpha)}) $$
$$
\iff\; \sphericalangle(\psi(x),\psi(y)) = \gamma_0.
$$
This completes our proof.
\end{proof}

Let $V$ be a vector space over $\mathbb{K} \in \{\R,\C\}$.
A transformation $A\colon V\to V$ is called semilinear if there exists a field automorphism $\sigma\colon\mathbb{K}\to \mathbb{K}$ such that we have
$$
A(\lambda_1 x_1 + \lambda_2 x_2) = \sigma(\lambda_1)Ax_1 + \sigma(\lambda_2)Ax_2 \quad (x_1,x_2\in V, \lambda_1,\lambda_2\in\mathbb{K}).
$$
If $\mathbb{K} = \R$, then it is well-known that the only field homomorphism is the identity, therefore semilinear maps are always linear.
In the complex case the identity and the conjugation are trivial field automorphisms, however it is a well-known fact that there are several others (see \cite{Ke}).

Before we prove Theorem \ref{Emain}, we need Uhlhorn's theorem for real inner product spaces, which, as far as we know, was not published anywhere in this generality.
However, we note that it is usual to prove Uhlhorn's theorem in Hilbert spaces with the fundamental theorem of projective geometry.
We will also follow this method, but with using an easy observation which works without completeness of the space.
Since later we will need the complex version as well, here we prove the real and complex cases together.

\begin{lem}[Uhlhorn's theorem in inner product spaces]\label{Ulem}
Let $H$ be a real or complex inner product space with $\dim H \geq 3$.
Let $\phi\colon P(H) \to P(H)$ be a bijective map which saisfies
$$
[u] \perp [v]\; \iff \; \phi([u]) \perp \phi([v]) \quad ([u],[v]\in P(H)).
$$
Then $\phi$ is a Wigner symmetry.
\end{lem}

\begin{proof}
Let $[x], [y] \in P(H)$ be two different lines. 
We will us use the notation $P_{[x],[y]} = \{[w]\in P(H)\colon [w]\subseteq [x,y]\}$.
It is clear that we have $P_{[x],[y]} \subseteq (\{[x], [y]\}^\perp)^\perp$.
Let us assume that $[z] \in (\{[x], [y]\}^\perp)^\perp \setminus P_{[x],[y]}$.
This means that the representing vectors $x$, $y$ and $z\in H$ are linearly independent.
By the Gramm-Schmidt orthogonalization theorem there are scalars $\lambda, \mu, \nu$ such that $0 \neq \lambda x + \mu y + \nu z$ and $[\lambda x + \mu y + \nu z] \in \{[x], [y]\}^\perp$.
But since $[x],[y],[z] \in (\{[x], [y]\}^\perp)^\perp$, we obtain that 
$$
[\lambda x + \mu y + \nu z] \in (\{[x], [y]\}^\perp)^\perp \subseteq \{[\lambda x + \mu y + \nu z]\}^\perp,
$$
which is a contradiction.
Therefore we obtain 
$$
P_{[x],[y]} = (\{[x], [y]\}^\perp)^\perp \qquad ([x], [y]\in P(H), [x]\neq [y]).
$$

The above observations and the properties of $\phi$ implies the following:
$$
\phi\left(P_{[x],[y]}\right) = \phi\left(\left(\{[x], [y]\}^\perp\right)^\perp\right) = \left(\phi\left(\{[x], [y]\}^\perp\right)\right)^\perp $$
$$
= \left(\{\phi([x]), \phi([y])\}^\perp\right)^\perp = P_{\phi([x]), \phi([y])}.
$$
Hence the fundamental theorem of projective geometry (see e.g. \cite{FTPG,Fa}) gives us a bijective semilinear transformation $A\colon H \to H$, with a field automorphism $\sigma\colon\mathbb{K}\to\mathbb{K}$, such that we have
\begin{equation}\label{semi}
\phi([v]) = [Av] \quad ([v]\in P(H)).
\end{equation}
Moreover, the semilinear transformation $A$ preserves orthogonality of vectors in both directions.

We will only deal with the complex case, since the real case can be shown similarly, even with some simplifications. 
So assume that $\mathbb{K} = \C$.
Let $x\in H$, $\|x\| = 1$ be a vector.
Clearly for any $t>0$, $A$ satisfies \eqref{semi} if and only if $\tfrac{1}{t} A$ does.
Therefore, without loss of generality, we may assume $\|Ax\|=1$.
Let $M$ be an arbitrary two-dimensional subspace of $H$ containing $x$, and consider the restricted semilinear map $A|_M$.
We also consider two linear isometries $W_1\colon \C^2 \to M$ and $W_2\colon AM \to \C^2$ such that $W_1(1,0) = x$ and $W_2(Ax) = (1,0)$.
The composition $W_2\circ A\circ W_1 \colon \C^2 \to \C^2$ is a semilinear map with the same field automorphism as $A$.
We also have $W_2\circ A\circ W_1 (1,0) = (1,0)$.
It is straightforward that there exists a linear transformation $Q\colon \C^2 \to \C^2$ such that we have 
$$
(W_2\circ A\circ W_1) (z_1,z_2) = Q (\sigma(z_1),\sigma(z_2)) \qquad ((z_1,z_2)\in \C^2).
$$
In particular, $Q(1,0) = (1,0)$.
Moreover, $W_2\circ A\circ W_1$ preserves orthogonality of vectors in both directions.
Let $\Q[i] = \{p+iq\colon p,q\in\Q\}$.
It is easy to see that $\sigma|_{\Q[i]^2}$ is either the conjugation or the identity.
Therefore $Q(1,0) \perp Q(0,1)$, furthermore $Q(1,0) - Q(0,1) = Q(1,-1) \perp Q(1,1) = Q(1,0) + Q(0,1)$.
This latter one is only possible if we have $\|Q(0,1)\| = \|Q(1,0)\| = \|(1,0)\| = 1$, therefore $Q$ is a linear isometry.

Finally, we consider 
$$
Q^{-1}\circ W_2\circ A\circ W_1 \colon \C^2\to\C^2, \quad (z_1,z_2) \mapsto (\sigma(z_1),\sigma(z_2)).
$$
This map preserves orthogonality of vectors in both directions.
Therefore we can calculate the following for every $r\geq 0$ and $s\in\R$:
$$
0 = \left\langle\left(r,r(\cos s + i\sin s)\right);\left(-r,r(\cos s + i\sin s)\right)\right\rangle 
$$
$$
= \left\langle\left(\sigma(r),\sigma(r(\cos s + i\sin s))\right);\left(-\sigma(r),\sigma(r(\cos s + i\sin s))\right)\right\rangle 
$$
$$
= -|\sigma(r)|^2 + \left|\sigma(r(\cos s + i\sin s))\right|^2.
$$
This implies $\left|\sigma(r(\cos s + i\sin s))\right| = |\sigma(r)|$.
Therefore, for every $t\in\R$ we calculate either
$$
|\sigma(t) + i| = |\sigma(t+i)| = |\sigma(t-i)| = |\sigma(t) - i|,
$$
if $\sigma|_{\Q[i]^2}$ is the identity, or
$$
|\sigma(t) + i| = |\sigma(t-i)| = |\sigma(t+i)| = |\sigma(t) - i|,
$$
if $\sigma|_{\Q[i]^2}$ is the conjugation.
But this implies $\sigma(t)\in\R$.
We conclude $\sigma(\R)\subseteq\R$, which implies that $\sigma|_\R$ is the identity, and therefore $\sigma$ is either the conjugation map or the identity.
From this we obtain that the restriction $A|_M$ is a linear or a conjugatelinear isometry for every two-dimensional subspace $M$ containing $x$.
Therefore $A$ is a linear or conjugate linear isometry on $H$ which is our desired conclusion.
\end{proof}

Now, we are in a position to verify our main result about transformations on unit spheres.

\begin{proof}[Proof of Theorem \ref{Emain}]
The $\alpha = \tfrac{\pi}{2}$ case is a direct consequence Lemma \ref{Ulem}. 
We will consider several possibilities separately in order to handle the remaining case.

\smallskip

\textbf{Case 1:} when $\alpha \leq \tfrac{\pi}{4}$.
Then by Lemmas \ref{multiple-sphericalangle-eq-leq_lem} and \ref{main-sphere_lem}, our map $\psi$ preserves the angles $2\alpha$ and $\beta := \arccos\left( \frac{4\cos^2\alpha}{\cos\alpha + 1} - 1 \right) \in (\alpha,2\alpha)$ in both directions.
Since $\alpha + \beta < 2\alpha + \beta < 4\alpha < \pi$, an application of Lemma \ref{sphere-diff-sum_lem} gives that $\psi$ preserves both of the (positive) angles $2\alpha-\beta$ and $\beta-\alpha$ in both directions.
Clearly, one of these angles is less than or equal to $\tfrac{\alpha}{2}$.
Therefore an iteration gives us a decreasing sequence of positive angles converging to zero such that $\psi$ preserves all of these angles in both directions.
Applying Lemma \ref{small-sphericalangles_lem} completes the proof of this case.

\smallskip

\textbf{Case 2:} when $\tfrac{\pi}{4} < \alpha < \tfrac{\pi}{2}$.
We will use the notation 
$$
\beta(\alpha) := \arccos\left( \frac{4\cos^2\alpha}{\cos\alpha + 1} - 1 \right) \quad \left(\alpha\in\left[0,\tfrac{\pi}{2}\right]\right).
$$
The following estimation shows that the continuous function $\beta(\alpha) - \alpha$ is strictly increasing on $\left[0,\tfrac{\pi}{2}\right]$:
$$
\frac{d\beta(\alpha)}{d\alpha}
=
\frac{4 \sin\alpha\cos\alpha (\cos\alpha+2)}{(\cos\alpha+1)^2\sqrt{1-\left(\frac{4 \cos^2\alpha}{\cos\alpha+1}-1\right)^2}}
$$
$$
= \frac{4 \sin \alpha\cos\alpha(\cos (\alpha )+2)}{(\cos\alpha+1)^2\sqrt{\frac{4 \cos^2\alpha}{\cos\alpha+1}}\sqrt{2-\frac{4 \cos^2\alpha}{\cos\alpha+1}}} 
= 
\frac{4 \sin\alpha\cos\alpha(\cos\alpha+2)}{(\cos\alpha+1) 2\cos\alpha\sqrt{2+2\cos\alpha-4 \cos^2\alpha}} 
$$
$$
> 
\frac{4 \sin\alpha\cos\alpha(\cos\alpha+2)}{(\cos\alpha+1) 2\cos\alpha\sqrt{2+2\cos\alpha-4 \cos^2\alpha}}
$$
$$
=
\frac{\cos\alpha + 2}{\cos\alpha + 1}
=
1 + \frac{1}{\cos\alpha + 1} > 1 \qquad (0 < \alpha < \tfrac{\pi}{2}).
$$
We also know that $\beta(\alpha) - \alpha < \alpha$ holds for every $\alpha\in (0,\tfrac{\pi}{2})$, and that $\beta(\alpha) - \alpha = \alpha$ for $\alpha \in \{0,\tfrac{\pi}{2}\}$.
Let us define a sequence $\{\alpha_n\}_{n=1}^\infty$ recursively as follows: $\alpha_1 := \tfrac{\pi}{4}$, $\tfrac{\pi}{4} < \alpha_n < \tfrac{\pi}{2}$ and $\beta(\alpha_n) - \alpha_n = \alpha_{n-1}$ $(n\in\N, n\geq 2)$.
This obviously defines a strictly increasing sequence of numbers which is therefore convergent: $\widehat\alpha := \lim_{n\to\infty} \alpha_n$.
Hence we get $\beta(\widehat\alpha) - \widehat\alpha = \widehat\alpha$, from which we conclude $\widehat\alpha = \tfrac{\pi}{2}$.

Next, let us examine the equation $2\pi-\beta(\alpha)-\alpha = 2\alpha$.
Since the left-hand side is strictly decreasing in $\alpha \in \left[0,\tfrac{\pi}{2}\right]$, and the right-hand side is strictly increasing, we conclude that there is a unique solution in the interval $\left[0,\tfrac{\pi}{2}\right]$.
A numerical calculation gives the following:
$$
2\pi-\beta(1.28)-1.28 > 2.59 > 2.56 = 2\cdot 1.28,
$$
and
$$
2\pi-\beta(1.29)-1.29 < 2.57 < 2.58 = 2\cdot 1.29.
$$
Therefore the unique solution, which will be denoted by $\check\alpha$, satisfies $1.28<\check\alpha<1.29$

We want to show for every $n\in\N$ that if $\alpha \leq \alpha_n$ holds, then $\psi$ is a $\sphericalangle$-isometry.
In fact, the previous case implies this for $n=1$.
Let us assume that we have verified this for an $n\in\N$, and let us consider an $\alpha \in (0,\alpha_{n+1}]$.
On one hand, if $\beta(\alpha)+\alpha \leq \pi$, then by Lemma \ref{sphere-diff-sum_lem} $\psi$ preservers the angle $\beta(\alpha)-\alpha$ in both directions, and since $\beta(\alpha)-\alpha \leq \alpha_n$, we are done by the inductional hypothesis.
On the other hand, if $\beta(\alpha)+\alpha > \pi$, then by Lemma \ref{sphere-0-1-element_lem} we have:
$$
\sphericalangle(x,y) \in \{\beta(\alpha)-\alpha, 2\pi-\beta(\alpha)-\alpha\} 
\;\iff\; \card\left(x^{(\alpha)}\cap y^{(\beta(\alpha))}\right) = 1
$$
$$
\iff\; \card\left(\psi(x)^{(\alpha)}\cap \psi(y)^{(\beta(\alpha))}\right) = 1 
$$
$$
\iff\; \sphericalangle(\psi(x),\psi(y)) \in \{\beta(\alpha)-\alpha, 2\pi-\beta(\alpha)-\alpha\}
$$
We distinguish two possibilities.
First, if $2\pi-\beta(\alpha)-\alpha \geq 2\alpha$, or equivalently when $\alpha\leq\check\alpha$, then we have
$$
\sphericalangle(x,y) = \beta(\alpha)-\alpha \;\iff\; \sphericalangle(x,y) \in \{\beta(\alpha)-\alpha, 2\pi-\beta(\alpha)-\alpha\} \text{ and } \sphericalangle(x,y) < 2\alpha,
$$
thus \eqref{multiple-angle-eq_eq} and \eqref{multiple-angle-leq_eq} of Lemma \ref{multiple-sphericalangle-eq-leq_lem} implies that $\psi$ preserves the angle $\beta(\alpha)-\alpha$ in both directions.
Since $\beta(\alpha)-\alpha \leq \alpha_{n}$, the map $\psi$ has to be a $\sphericalangle$-isometry.
Second, if $2\pi-\beta(\alpha)-\alpha < 2\alpha$, or equivalently $\alpha>\check\alpha$, then we have $2\pi-4\alpha < 2\pi-4\check\alpha < 1.2 < \check\alpha$.
Lemma \ref{sphere-0-1-element_lem} gives the following:
$$
\sphericalangle(x,y) = 2\pi-4\alpha \;\iff\; \card\left(x^{(2\alpha)}\cap y^{(2\alpha)}\right) = 1 \quad (x,y\in S_H).
$$
Since $\psi$ preserves the angle $2\alpha$ in both directions by Lemma \ref{multiple-sphericalangle-eq-leq_lem}, a straightforward argument shows that $\psi$ also preserves the angle $2\pi-4\alpha$ in both directions.
Therefore, by the previous possibility and the inductional hypothesis, we infer the $\sphericalangle$-isometriness of $\psi$.

Finally, we easily conclude that whenever $\alpha < \widehat\alpha = \tfrac{\pi}{2}$, then $\psi$ is a $\sphericalangle$-isometry, which completes the present case.

\smallskip

\textbf{Case 3:} when $\tfrac{3\pi}{4} < \alpha < \pi$. 
Then by Lemma \ref{sphere-0-1-element_lem} we have $\card(x^{(\alpha)} \cap y^{(\alpha)}) = 1$ if and only if $\sphericalangle(x,y) = 2\pi - 2\alpha$.
Therefore $\psi$ preserves the angle $2\pi-2\alpha \in \left(0,\tfrac{\pi}{2}\right)$ in both directions, and the previous cases complete the present one.

\smallskip

\textbf{Case 4:} when $\alpha = \tfrac{3\pi}{4}$.
Similarly as above, we get that $\psi$ preserves the angle $\tfrac{\pi}{2}$ in both directions.
By Lemma \ref{sphere-0-1-element_lem} we have the following property:
$$
\sphericalangle(x,y) \in \left\{\tfrac{\pi}{4},\tfrac{3\pi}{4}\right\} \;\iff\; \sphericalangle(\psi(x),\psi(y)) \in \left\{\tfrac{\pi}{4},\tfrac{3\pi}{4}\right\} \quad (x,y\in S_H).
$$
But the angle $\tfrac{3\pi}{4}$ is preserved in both directions, whence we get that the same holds for the angle $\tfrac{\pi}{4}$.
Therefore, by Case 1, $\psi$ is a $\sphericalangle$-isometry.

\smallskip

\textbf{Case 5:} when $\tfrac{\pi}{2} < \alpha < \tfrac{3\pi}{4}$ and $\alpha \neq \tfrac{2\pi}{3}$.
We set $\alpha_1 := \tfrac{5\pi}{8}$ and define $\alpha_n$ for $n \geq 2$ by the recursive formula $\alpha_n = \pi - \tfrac{\alpha_{n-1}}{2}$ $(n\geq 2, n\in\N)$.
We observe that $\alpha_{2k-1} \in [\tfrac{5\pi}{8},\tfrac{2\pi}{3})$, $\alpha_{2k} \in (\tfrac{2\pi}{3},\tfrac{11\pi}{16}]$ $(k\in\N)$ are satisfied.
Moreover, one easily obtains the explicit formula $\alpha_n = (-1)^{n-1}\tfrac{\alpha_1}{2^{n-1}} + \sum_{j=0}^{n-2} (-1)^j \tfrac{\pi}{2^j}$ $(n\geq 2)$, from which we infer $\lim_{n\to\infty} \alpha_n = \tfrac{2\pi}{3}$.

We show that if $\tfrac{\pi}{2} < \alpha \leq \alpha_n$ and $n$ is odd, or $\alpha_n \leq \alpha < \tfrac{3\pi}{4}$ and $n$ is even $(n\in\N)$, then $\psi$ is a $\sphericalangle$-isometry.
In order to do this, we use induction.
If $n=1$, then $\tfrac{\pi}{2} < \alpha \leq \alpha_{1} = \tfrac{5\pi}{8}$ implies that $\psi$ preserves the angle $2\pi-2\alpha \in \left[\tfrac{3\pi}{4},\pi\right)$. 
From the previous cases it readily follows that $\psi$ is a $\sphericalangle$-isometry.
Let us suppose that the claim has been proven for some $n\geq 1$, and let us investigate it for $n+1$.
On one hand, if $n = 2k-1$ with some $k\in\N$, then $\alpha \geq \alpha_{2k} > \tfrac{2\pi}{3}$ implies that $\psi$ preserves the angle $2\pi-2\alpha \in (0,\alpha_{2k-1}]$ in both directions.
Thus, by the inductional hypothesis, either $\psi$ is a $\sphericalangle$-isometry, or $2\pi-2\alpha = \tfrac{\pi}{2}$.
But in the latter case we have $\alpha = \tfrac{3\pi}{4}$, which contradicts to our assumption, therefore $\psi$ is a $\sphericalangle$-isometry.
On the other hand, if $n = 2k$ with some $k\in\N$, then $\tfrac{\pi}{2} < \alpha \leq \alpha_{2k+1}$ implies that $\psi$ preserves the angle $2\pi-2\alpha \geq \alpha_{2k}$ in both directions.
The inductional hypothesis and the previous cases togehter yield that $\psi$ is a $\sphericalangle$-isometry.

From the properties of the sequence $\{\alpha_n\}_{n=1}^\infty$ it follows that whenever $\alpha \in (\tfrac{\pi}{2}, \tfrac{3\pi}{4})\setminus \{\tfrac{2\pi}{3}\}$, then $\psi$ is a $\sphericalangle$-isometry.

\smallskip

\textbf{Case 6:} when $\alpha = \tfrac{2\pi}{3}$.
By Lemma \ref{sphere-0-1-element_lem} for every $x,y\in S_H$ we have
$$
\sphericalangle(x,y) > \tfrac{2\pi}{3}
\;\iff\; x^{(2\pi/3)} \cap y^{(2\pi/3)} = \emptyset
$$
$$
\;\iff\; \psi(x)^{(2\pi/3)} \cap \psi(y)^{(2\pi/3)} = \emptyset
\;\iff\; \sphericalangle(\psi(x),\psi(y)) > \tfrac{2\pi}{3}.
$$
It is straightforward that for every $u\in x^{(2\pi/3)}$ we have $\card(x^{(2\pi/3)} \cap u^{(2\pi/3)}) = 1$.
We will denote this unique element by $\widetilde{u}_x$, which depends only on $x$ and $u$.
We easily obtain the following
$$
\left\{\psi(\widetilde{u}_x)\right\}
 = \psi\left(x^{(2\pi/3)} \cap u^{(2\pi/3)}\right)
$$
$$
 = \psi(x)^{(2\pi/3)} \cap \psi(u)^{(2\pi/3)}
 = \left\{\widetilde{\psi(u)}_{\psi(x)}\right\}
\quad \left(x\in S_H, u\in x^{(2\pi/3)}\right).
$$

We claim that $\sphericalangle(x,y) = \pi$ holds if and only if $\sphericalangle(x,y) > \tfrac{2\pi}{3}$ and for every $u\in x^{(2\pi/3)}$ and $z\in y^{(2\pi/3)} \cap u^{(2\pi/3)}$ we have $\sphericalangle(\widetilde u_x, \widetilde{z}_y) = \tfrac{2\pi}{3}$.
For the necessity part, we assume that $\sphericalangle(x,y) = \pi$.
We consider the following isometry $T\colon H\to H$, $T(\lambda x + h) = \lambda x - h$ $(h\in H, h\perp x)$.
Clearly, we have $\widetilde{u}_x = Tu$ $(u\in x^{(2\pi/3)})$, $\widetilde{z}_y = Tz$ $(z\in y^{(2\pi/3)})$, $T (x^{(2\pi/3)}) = x^{(2\pi/3)}$ and $T (y^{(2\pi/3)}) = y^{(2\pi/3)}$.
If $u\in x^{(2\pi/3)}$ and $z\in y^{(2\pi/3)} \cap u^{(2\pi/3)}$, then $\widetilde{z}_y = Tz \in (Ty)^{(2\pi/3)} \cap (Tu)^{(2\pi/3)} = y^{(2\pi/3)} \cap (\widetilde{u}_x)^{(2\pi/3)}$, which gives us $\sphericalangle(\widetilde u_x, \widetilde{z}_y) = \tfrac{2\pi}{3}$.

For the sufficiency part, we assume that $\tfrac{2\pi}{3} < \sphericalangle(x,y) < \pi$, and we would like to conclude that there exist $u\in x^{(2\pi/3)}$ and $z\in y^{(2\pi/3)} \cap u^{(2\pi/3)}$ such that $\sphericalangle(\widetilde u_x, \widetilde{z}_y) \neq \tfrac{2\pi}{3}$.
Let $\gamma = \sphericalangle(x,y) \in \left(\tfrac{2\pi}{3},\pi\right)$.
In order to do this, we choose an orthonormal system $\{e_1,e_2\}$ such that $x= e_1$, $y = \cos\gamma \cdot e_1 + \sin\gamma \cdot e_2$.
We define $y' := \sin\gamma \cdot e_1 - \cos\gamma \cdot e_2$ and fix $u := -\tfrac{1}{2} \cdot e_1 - \tfrac{\sqrt{3}}{2} \cdot e_2 \in x^{(2\pi/3)}$.
We have $\widetilde{u}_x = -\tfrac{1}{2} \cdot e_1 + \tfrac{\sqrt{3}}{2} \cdot e_2$.
A straightforward calculation gives that for any $z\in y^{(2\pi/3)}$ there exists a $\delta \in [0,2\pi)$ and an $e_3\in S_H$, $e_3\perp e_j$ $(j=1,2)$ such that
$$
z = -\tfrac{1}{2} \cdot y + \tfrac{\sqrt{3}}{2} \cdot (\cos\delta \cdot y' + \sin\delta \cdot e_3),
$$
and of course we have 
$$
\widetilde{z}_y = -\tfrac{1}{2} \cdot y - \tfrac{\sqrt{3}}{2} \cdot (\cos\delta \cdot y' + \sin\delta \cdot e_3).
$$
We have $z \in y^{(2\pi/3)} \cap u^{(2\pi/3)}$ if and only if the following equation is satisfied with the above $\delta$:
\begin{equation}\label{u-ip-z_eq}
\langle u,z \rangle = \tfrac{1}{4}\cos\gamma + \tfrac{\sqrt{3}}{4}\sin\gamma - \tfrac{\sqrt{3}}{4}\sin\gamma\cos\delta + \tfrac{3}{4}\cos\gamma\cos\delta = -\tfrac{1}{2}.
\end{equation}
Furthermore, we have $\sphericalangle(\widetilde u_x, \widetilde{z}_y) = \tfrac{2\pi}{3}$ exactly when
$$
\langle \widetilde u_x, \widetilde{z}_y \rangle = \tfrac{1}{4}\cos\gamma - \tfrac{\sqrt{3}}{4}\sin\gamma + \tfrac{\sqrt{3}}{4}\sin\gamma\cos\delta + \tfrac{3}{4}\cos\gamma\cos\delta = -\tfrac{1}{2}.
$$
These two equations gives $\tfrac{\sqrt{3}}{2}\sin\gamma - \tfrac{\sqrt{3}}{2}\sin\gamma\cos\delta = 0$,
which implies $\cos\delta = 1$.
Therefore from \eqref{u-ip-z_eq} we get $\tfrac{2\pi}{3} = \gamma$, which contradicts to our assumptions.
Therefore our claim is verified.

Now, by these observations we infer the following equivalence-chain:
$$
\sphericalangle(x,y) = \pi 
$$
$$
\iff\; \sphericalangle(x,y) > \tfrac{2\pi}{3} \; \text{and} \; \sphericalangle(\widetilde u_x, \widetilde{z}_y) = \tfrac{2\pi}{3} \quad \left(u\in x^{(2\pi/3)}, z\in y^{(2\pi/3)} \cap u^{(2\pi/3)}\right)
$$
$$
\;\iff\; \left\{ \begin{matrix}
\sphericalangle(\psi(x),\psi(y)) > \tfrac{2\pi}{3} \; \text{and} \; \sphericalangle(\widetilde{\psi(u)}_{\psi(x)}, \widetilde{\psi(z)}_{\psi(y)}) = \tfrac{2\pi}{3}\\
\left(\psi(u)\in \psi(x)^{(2\pi/3)}, \psi(z)\in \psi(y)^{(2\pi/3)} \cap \psi(u)^{(2\pi/3)}\right)
\end{matrix} \right.
$$
$$
\;\iff\; \left\{ \begin{matrix}
\sphericalangle(\psi(x),\psi(y)) > \tfrac{2\pi}{3} \; \text{and} \; \sphericalangle(\widetilde{w}_{\psi(x)}, \widetilde{s}_{\psi(y)}) = \tfrac{2\pi}{3}\\
\left(w\in \psi(x)^{(2\pi/3)}, s\in \psi(y)^{(2\pi/3)} \cap w^{(2\pi/3)}\right)
\end{matrix} \right.
\;\iff\; \sphericalangle(\psi(x),\psi(y)) = \pi.
$$
Since $\psi$ preserves the angles $\tfrac{2\pi}{3}$ and $\pi$ in both directions, it also preserves the angle $\tfrac{\pi}{3}$ in both directions. 
Finally, Case 2 completes the proof.
\end{proof}


\section{Symmetries of projective spaces}\label{proj_sec}

The aim of this section is to prove our main results on transformations of projective spaces namely, Theorems \ref{Umain-real}, \ref{Umain-complex-2d} and \ref{Umain-complex}.
We begin with the verification of Theorem \ref{Umain-complex-2d} which is a consequence of Theorem \ref{Emain}.

\begin{proof}[Proof of Theorem \ref{Umain-complex-2d}]
By Bloch's represantion (see e.g. \cite{Bloch}), elements of $P(\C^2)$ can be represented as points on the unit sphere $S_{\R^3}$ of $\R^3$ in the following way:
\begin{equation}\label{qubit-repr}
\rho\colon P(\C^2)\to S_{\R^3}, \quad \rho\left([(\cos\theta, e^{i\nu}\sin\theta)]\right) = (\sin 2\theta \cos\nu, \sin 2\theta \sin\nu, \cos 2\theta),
\end{equation}
where $\theta \in \left[0,\tfrac{\pi}{2}\right]$, $\nu \in [0,2\pi)$.
Furthermore, we have
$$
\sphericalangle\left(\rho([u]),\rho([v])\right) = 2\cdot\measuredangle([u],[v]) \quad ([u],[v]\in P(H)).
$$
Now, we easily obtain that $\rho\circ\phi\circ\rho^{-1}\colon S_{\R^3}\to S_{\R^3}$ preserves the angle $2\alpha$ in both directions.
On one hand, if $\alpha \neq \tfrac{\pi}{4}$, then by Theorem \ref{Emain} we obtain that $\rho\circ\phi\circ\rho^{-1}$ is a $\sphericalangle$-isometry, and thus $\phi$ is a Wigner symmetry.
On the other hand, in case when $\alpha = \tfrac{\pi}{4}$, then from Theorem \ref{Emain} we get that there is a $\sphericalangle$-isometry $\psi\colon S_{\R^3}\to S_{\R^3}$ and a map $\varepsilon\colon S_{\R^3} \to \{-1,1\}$ such that $\rho\circ\phi\circ\rho^{-1} = \varepsilon\cdot\psi$.
Therefore we have $\phi = \rho^{-1}\circ(\varepsilon\cdot\psi)\circ\rho = \epsilon\circ\rho^{-1}\circ\psi\circ\rho$ where $\epsilon \colon P(\C^2)\to P(\C^2)$, $\epsilon([v]) \in \{[v],[v]^\perp\}$ $([v]\in P(H))$ and $\rho^{-1}\circ\psi\circ\rho$ is a Wigner symmetry.
This completes the proof.
\end{proof}

We proceed with verifying some lemmas.
In the first one we show that $\measuredangle$ defines a metric on $P(H)$.

\begin{lem}\label{angle-metric_lem}
Let $H$ be a real or complex inner product space.
Then the angle $\measuredangle$ defines a metric on $P(H)$.
\end{lem}

\begin{proof} 
The $\dim H < 2$ case is obvious.
If $\dim H = 2$, then this is quite trivial in the real case.
For the complex case, we utilize Bloch's representation and the fact that $\sphericalangle$ gives a metric on $S_{\R^3}$.

Now, let us consider the general case, i.e. when $\dim H > 2$.
Let $u,v,w\in H$ with $\|u\| = \|v\| = \|w\| = 1$. 
If $\dim [u,v,w] \leq 2$, then we are done by the two-dimensional case.
Otherwise, let $\{e_1,e_2,e_3\}$ be an orthonormal system in $H$ such that $[u] = [e_1]$, $v\in[e_1,e_2]$ and $w\in [e_1,e_2,e_3]$. 
We may assume without loss of generality, by replacing $e_j$'s by their scalar multiples, that we have $[v] = [\cos \alpha\cdot e_1 + \sin\alpha \cdot e_2]$ and $[w] = [\cos\beta \cdot e_1 + \mu \sin\beta \cos\gamma \cdot e_2 + \sin\beta \sin\gamma \cdot e_3]$ with some $\alpha \in \left(0,\tfrac{\pi}{2}\right]$, $\beta \in \left(0,\tfrac{\pi}{2}\right]$, $\gamma \in \left(0,\tfrac{\pi}{2}\right]$ and $|\mu| = 1$. 
We easily compute $\measuredangle([u],[v]) = \alpha$ and $\measuredangle([u],[w]) = \beta$.

On one hand, if $\alpha+\beta \geq \tfrac{\pi}{2}$, then we obviously have $\measuredangle([v],[w]) \leq \measuredangle([u],[v]) + \measuredangle([u],[w])$. 
On the other hand, if $\alpha+\beta < \tfrac{\pi}{2}$, then we have $\cos\alpha \cos\beta > \sin\alpha \sin\beta$. 
Therefore we can estimate in the following way:
\begin{equation}
\begin{gathered}
\measuredangle([v],[w]) = \arccos|\langle v,w\rangle| = \arccos|\cos \alpha \cos\beta + \overline{\mu} \sin\alpha \sin\beta \cos\gamma|
\\
\leq \arccos(\cos\alpha \cos\beta - \sin\alpha \sin\beta \cos\gamma)
\\
\leq \arccos(\cos\alpha \cos\beta - \sin\alpha \sin\beta) = \alpha + \beta = \measuredangle([u],[v]) + \measuredangle([u],[w]),
\end{gathered}
\end{equation}
which verifies the triangle inequality.
\end{proof}

We shall use the notation 
$$
\mathcal{A}^\alpha = \{ [u]\in P(H) \colon \measuredangle([u],[v]) = \alpha, \; \forall \; [v]\in\mathcal{A} \},
$$
where $\mathcal{A} \subset P(H)$ and $\alpha \in (0,\tfrac{\pi}{2}]$. 
If $\mathcal{A} = \{[v]\}$ then we will simply write $[v]^\alpha$ instead of $\{[v]\}^\alpha$.
Let us note that the previously defined set $\mathcal{A}^\perp$ coincides with $\mathcal{A}^{\pi/2}$.
We will denote the set of numbers in $\K \in \{\C,\R\}$ with unit modulus by $\irC$.
In the following lemma the dimensionality assumption is crucial.

\begin{lem}\label{proj-0-1-element_lem}
Let $H$ be a real or complex inner product space with $\dim H \geq 3$.
Assume that $0 < \alpha \leq \beta \leq \tfrac{\pi}{2}$, and $v,w\in H$ with $\|v\| = \|w\| = 1$ and $\gamma := \measuredangle([v],[w]) \neq 0$. 
Then we have
\begin{equation}\label{proj-0-1-element_eq}
\begin{gathered}[]
[v]^\alpha\cap[w]^\beta = \emptyset \;\iff\; \gamma < \beta - \alpha \text{ or } \gamma > \alpha + \beta,\\
\card([v]^\alpha\cap[w]^\beta) = 1 \;\iff\; \left\{\begin{matrix} \gamma = \beta - \alpha, \text{ or } \gamma = \alpha + \beta < \tfrac{\pi}{2}, \\ \text{ or } \dim H = 3 \text{ and } \alpha = \beta = \tfrac{\pi}{2}. \end{matrix}\right.
\end{gathered}
\end{equation}
Moreover, if $\gamma = \tfrac{\pi}{2} = \alpha + \beta$, then $[v]^\alpha\cap[w]^\beta = \{[\cos\alpha \cdot v + \lambda\sin\alpha \cdot w] \colon \lambda\in\irC\}$.
\end{lem}

\begin{proof}
If $\alpha = \tfrac{\pi}{2}$, then we also have $\beta = \tfrac{\pi}{2}$.
Clearly, we have $\card([v]^{\pi/2}\cap[w]^{\pi/2}) = 1$ if and only if $\dim H = 3$.
Therefore from now on we may assume $0 < \alpha < \tfrac{\pi}{2}$.
The first possibility in \eqref{proj-0-1-element_eq} is obvious by the triangle inequality, therefore we may assume throughout the proof that $\beta - \alpha \leq \gamma \leq \alpha + \beta$.
We fix an orthonormal system $\{e_1,e_2\}$ such that $[v] = [e_1]$ and $[w] = [\cos\gamma \cdot e_1 + \sin\gamma \cdot e_2]$. 
Let $[u]\in [v]^\alpha$, then
\begin{equation}\label{form-of-u_eq}
[u] = [\cos\alpha\cdot e_1 + \lambda \sin\alpha \cos\delta\cdot e_2 + \mu \sin\alpha \sin\delta \cdot e_3]
\end{equation}
holds with some $\delta\in[0,\tfrac{\pi}{2}]$, $|\lambda| = |\mu| = 1$, and $e_3$ such that $\|e_3\| = 1$, $e_3\perp e_j$ $(j=1,2)$.
Of course, the three coordinates in \eqref{form-of-u_eq} are all non-zero if and only if $0 < \delta < \frac{\pi}{2}$.

We have $[u]\in [v]^\alpha\cap[w]^\beta$ if and only if 
\begin{equation}\label{cosb_eq}
\cos \beta = |\cos\alpha\cos\gamma + \lambda\sin\alpha\sin\gamma\cos\delta|.
\end{equation}
Observe that if there is a solution $(\lambda,\delta)\in \irC \times (0,\tfrac{\pi}{2}]$ of \eqref{cosb_eq}, then choosing $\mu = -1$ and $1$ in \eqref{form-of-u_eq} gives us two different elements of $[v]^\alpha\cap[w]^\beta$.
Therefore we conclude that $\card([v]^\alpha\cap[w]^\beta) = 1$ implies that for every solution of \eqref{cosb_eq} we have $\delta = 0$.
In particular, we have $[u]\in P_{[v],[w]}$

Clearly, the map
$$
\vartheta \colon \irC \times [0,\tfrac{\pi}{2}] \to \K, \quad (\lambda,\delta) \mapsto \cos\alpha\cos\gamma + \lambda\sin\alpha\sin\gamma\cos\delta
$$
is injective on $\irC \times [0,\tfrac{\pi}{2})$. Furthermore, we have $\vartheta\big( \irC \times [0,\tfrac{\pi}{2}) \big) = \{z\in\K \colon 0 < |z-\cos\alpha\cos\gamma| \leq \sin\alpha\sin\gamma \}$ and $\vartheta\big( \irC \times \{\tfrac{\pi}{2}\} \big) = \{\cos\alpha\cos\gamma\}$.
Therefore $\card([v]^\alpha\cap[w]^\beta) > 0$ holds if and only if $\cos(\alpha+\gamma) \leq \cos\beta \leq \cos(\alpha-\gamma)$, which is equivalent to $|\alpha-\gamma| \leq \beta \leq \alpha+\gamma$. 
It is straightforward that this condition is fulfilled under our assumption, i.e. $\beta - \alpha \leq \gamma \leq \alpha + \beta$.
(In fact, these two conditions are equivalent).

Now, the following equivalence-chain is straightforward from the observations made above:
$$
\card([v]^\alpha\cap[w]^\beta) = 1 \;\iff\; (\lambda, \delta) = (-1,0) \text{ or } (1,0) \text{ is a unique solution for } \eqref{cosb_eq}
$$
$$
\iff \; \left\{\begin{matrix}
(\lambda,\delta) = (1,0) \text{ is a solution for } \eqref{cosb_eq} \text{ and } \cos\alpha\cos\gamma \neq 0, \text{ or } \\
(\lambda,\delta) = (-1,0) \text{ is a solution for } \eqref{cosb_eq} \text{ and } \cos\alpha\cos\gamma \geq \sin\alpha\sin\gamma \, (>0)
\end{matrix}\right.
$$
$$
\iff \; \left\{\begin{matrix}
(\lambda,\delta) = (1,0) \text{ is a solution for } \eqref{cosb_eq} \text{ and } \gamma\neq\tfrac{\pi}{2}, \text{ or } \\
(\lambda,\delta) = (-1,0) \text{ is a solution for } \eqref{cosb_eq} \text{ and } \alpha + \gamma \leq \tfrac{\pi}{2}
\end{matrix}\right.
$$
$$
\iff \; \cos\beta = \cos(\alpha-\gamma) \text{ and } \gamma\neq\tfrac{\pi}{2}, \text{ or } \cos\beta = \cos(\alpha+\gamma) \text{ and } \alpha + \gamma \leq \tfrac{\pi}{2}
$$
$$
\iff \; \beta = \alpha-\gamma \text{ and } \gamma\neq\tfrac{\pi}{2}, \text{ or } \beta = \gamma-\alpha \text{ and } \gamma\neq\tfrac{\pi}{2}, \text{ or } \beta = \alpha+\gamma
$$
$$
\iff \; \gamma = \alpha-\beta, \text{ or } \gamma = \alpha + \beta < \tfrac{\pi}{2}, \text{ or } \gamma = \beta - \alpha
$$
$$
\iff \; \gamma = \alpha + \beta < \tfrac{\pi}{2} \text{ or } \gamma = \beta-\alpha.
$$
This verifies \eqref{proj-0-1-element_eq}.

Finally, if $\gamma = \tfrac{\pi}{2} = \alpha + \beta$, then \eqref{cosb_eq} becomes
$$
\sin \alpha = \cos \beta = \sin\alpha\cos\delta.
$$
From this we infer $\delta = 0$, and by \eqref{form-of-u_eq} we get
$$
[v]^\alpha\cap[w]^\beta \subseteq [v]^\alpha \subseteq \{[\cos\alpha \cdot e_1 + \lambda\sin\alpha \cdot e_2] \colon \lambda\in\irC\} = \{[\cos\alpha \cdot v + \lambda\sin\alpha \cdot w] \colon \lambda\in\irC\}.
$$
But obviously, $\{[\cos\alpha \cdot v + \lambda\sin\alpha \cdot w] \colon \lambda\in\irC\} \subseteq [v]^\alpha\cap[w]^\beta$ is also fulfilled, which completes the proof.
\end{proof}

The following three lemmas are consequences of Lemma \ref{proj-0-1-element_lem} and are proven in a similar way as the lemmas in Section \ref{sphere_sec}.
Namely, in each of them we provide a characterization of a relative position in terms of some other relative positions which are preserved by our transformation in both directions, and therefore we conclude that this relative position is also preserved in both directions.

\begin{lem}\label{multiple-measuredangle-leq_lem}
Let $H$ be a real or complex inner product space with $\dim H \geq 3$, and $\phi\colon P(H) \to P(H)$ be a bijection.
If $\phi$ preserves the angle $\alpha \in (0,\tfrac{\pi}{4})$ in both directions, then $\phi$ shares the following property:
\begin{equation}\label{proj-multiple-angle-leq_eq}
\measuredangle([v],[w]) \leq j\alpha \;\iff\; \measuredangle(\phi([v]),\phi([w])) \leq j\alpha \quad \left([v],[w] \in P(H), \; 2\leq j < \tfrac{\pi}{2\alpha}\right)
\end{equation}
\end{lem}

\begin{proof}
Clearly, we have $\measuredangle([v],[w]) \leq j\alpha$ $\left(j\in\N, \; 2\leq j < \frac{\pi}{2\alpha}\right)$ if and only if there exists a $(j-1)$-element sequence $\{[u_l]\}_{l=1}^{j-1} \subset P(H)$ such that $\measuredangle([v],[u_1]) = \measuredangle([u_l],[u_{l+1}]) = \measuredangle([u_{j-1}],[w]) = \alpha$ $(l\in\N, 1\leq l \leq j-2)$.
Namely, one direction is immediate by Lemma \ref{angle-metric_lem}.
For the other one, by Bloch's representation we easily infer the existence of such a $(j-1)$-element sequence in $P_{[v],[w]}$.
Therefore, using the same technique as in Lemma \ref{multiple-sphericalangle-eq-leq_lem}, the property \eqref{proj-multiple-angle-leq_eq} is yielded.
\end{proof}

\begin{lem}\label{proj-diff-sum_lem}
Let $H$ be a real or complex inner product space with $\dim H \geq 3$, $0 < \alpha < \beta < \tfrac{\pi}{2}$, and $\phi\colon P(H) \to P(H)$ be a bijection.
If $\phi$ preserves the angles $\alpha$ and $\beta$ in both directions, then
\begin{itemize}
\item[(i)] $\phi$ also preserves the angles $\beta-\alpha$ and $\alpha + \beta$ in both directions if $\alpha + \beta < \tfrac{\pi}{2}$,
\item[(ii)] $\phi$ also preserves the angle $\beta-\alpha$ in both directions if $\alpha + \beta \geq \tfrac{\pi}{2}$,
\end{itemize}
\end{lem}

\begin{proof}
(i): An application of Lemma \ref{proj-0-1-element_lem} gives the following property:
\begin{equation}\label{proj-a-b_eq}
\measuredangle([v],[w]) \in \{\beta-\alpha, \alpha+\beta\} \iff \measuredangle(\phi([v]),\phi([w])) \in \{\beta-\alpha, \alpha+\beta\}  \;\; ([v],[w] \in P(H)).
\end{equation}
Let us observer that there exists a $j\in\N$, $2\leq j < \frac{\pi}{2\alpha}$ with $\beta-\alpha < j\alpha < \alpha+\beta$.
A technique similar to the one which was used in Lemma \ref{sphere-diff-sum_lem} completes the proof of this case.

(ii): A similar, but easier argument verifies this part.
\end{proof}

\begin{lem}\label{multiple-measuredangle_lem}
Let $H$ be a real or complex inner product space with $\dim H \geq 3$, and $\phi\colon P(H) \to P(H)$ be a bijection.
If $\phi$ preserves the angle $\alpha \in (0,\tfrac{\pi}{4})$ in both directions, then $\phi$ shares the following property:
\begin{equation}\label{proj-multiple-angle_eq}
\measuredangle([v],[w]) = j\alpha \;\iff\; \measuredangle(\phi([v]),\phi([w])) = j\alpha \quad \left([v],[w] \in P(H), \; 2\leq j < \frac{\pi}{2\alpha}\right)
\end{equation}
\end{lem}

\begin{proof}
For $j=2$, this is straightforward from Lemma \ref{proj-0-1-element_lem}.
For $j>2$, the statement can be verified using a simple recursion and (i) of Lemma \ref{proj-diff-sum_lem}.
\end{proof}

Next, we provide a counterpart of Lemma \ref{small-sphericalangles_lem}.
Since its verification can be done along the same lines as the proof of Lemma \ref{small-sphericalangles_lem}, we shall omit it.

\begin{lem}\label{small-measuredangles_lem}
Let $H$ be a real or complex inner product space with $\dim H \geq 3$, and $\phi\colon P(H) \to P(H)$ be a bijection.
If there is a sequence of positive angles $\{\alpha_n\}_{n=1}^\infty \subset (0,\tfrac{\pi}{2})$ such that $\lim_{n\to\infty} \alpha_n = 0$ and $\phi$ preserves these angles in both directions, then $\phi$ is a Wigner symmetry.
\end{lem}

If $H$ is a real or complex inner product space and $[v]$ is a line in it, then the symbol $H\ominus [v]$ will denote the set of those vectors in $H$ which are orthogonal to $v$.
This is a linear subspace, thus it can be considered as an inner product space with the restricted inner product.

Now, we are in the position to present the verification of Theorem \ref{Umain-real}.
In order to do this we will apply Theorem \ref{Emain}.

\begin{proof}[Proof of Theorem \ref{Umain-real}]
During our proof we will distinguish six different cases. 
In the first three of them we will deal with the possibility when $H$ is of dimension at least four, and in the last three ones we will handle the three-dimensional case.

\smallskip

\textbf{Case 1:} when $\dim H \geq 4$ and $0 < \alpha < \tfrac{\pi}{3}$.
Let $v \in S_H$ be arbitrary and $w \in S_H$ such that $\phi([v]) = [w]$ holds. 
Then we have $\phi([v]^\alpha) = \phi([v])^\alpha = [w]^\alpha$.
We consider a bijective isometry $R\colon H\to H$ such that $Rw = v$ and define the transformation $\widetilde\phi\colon P(H) \to P(H)$, $\widetilde\phi([u]) = R(\phi([u]))$ which is bijective and preserves the angle $\alpha$ in both directions.
In addition, we have $\widetilde\phi([v]) = [v]$, whence $\widetilde\phi([v]^\alpha) = [v]^\alpha$ follows.
It is straightforward that
$$
[v]^\alpha = \{[\cos\alpha\cdot v + \sin\alpha\cdot u] \colon u\in S_{H\ominus [v]}\},
$$
where $\dim (H\ominus [v]) \geq 3$.
Therefore we have a bijective map $\psi\colon S_{H\ominus [v]} \to S_{H\ominus [v]}$ which implements the restricted transformation $\widetilde\phi\big|_{[v]^\alpha}$, i.e. we have
$$
\widetilde\phi([\cos\alpha\cdot v + \sin\alpha\cdot u]) = [\cos\alpha\cdot v + \sin\alpha\cdot \psi(u)] \quad (u\in S_{H\ominus [v]}).
$$
We calculate the following for every $u_1, u_2 \in S_{H\ominus [v]}$:
\begin{equation}\label{convert-angle-from-proj-tosphere_eq}
\measuredangle\big([\cos\alpha\cdot v + \sin\alpha\cdot u_1],[\cos\alpha\cdot v + \sin\alpha\cdot u_2]\big) = \arccos \left(\left|\cos^2\alpha + \langle u_1, u_2\rangle \sin^2\alpha\right|\right).
\end{equation}
If $u_1$ and $u_2$ run through $S_{H\ominus [v]}$, then $\cos^2\alpha + \langle u_1, u_2\rangle \sin^2\alpha$ runs through $\left[\cos 2\alpha, 1\right]$.
Thus the above angle runs through the interval $\left[0, \min\left(\tfrac{\pi}{2}, 2\alpha\right)\right]$.

The inequality $-\cos\alpha < \cos 2\alpha$ implies that
$$
\measuredangle\big([\cos\alpha\cdot v + \sin\alpha\cdot u_1],[\cos\alpha\cdot v + \sin\alpha\cdot u_2]\big) = \alpha
$$ 
holds exactly when 
$$
\sphericalangle(u_1,u_2) = \arccos \left(\frac{\cos\alpha}{1+\cos\alpha}\right) \in \left(\tfrac{\pi}{3},\arccos\tfrac{1}{3}\right) \subset \left(\tfrac{\pi}{3},\tfrac{\pi}{2}\right).
$$
Since $\widetilde\phi$ preserves the angle $\alpha$ between lines in both directions, it is apparent that $\psi$ preserves the angle $\arccos \left(\tfrac{\cos\alpha}{1+\cos\alpha}\right)$ between unit vectors in both directions.
By Theorem \ref{Emain} we easily obtain that $\psi$ is a bijective $\sphericalangle$-isometry. 
Therefore if $[w_1], [w_2] \in [v]^\alpha$, then we have $\measuredangle([w_1],[w_2]) = \measuredangle\left(\widetilde\phi([w_1]),\widetilde\phi([w_2])\right)$, and thus we infer
$$
\measuredangle(\phi([w_1]),\phi([w_2])) = \measuredangle([w_1],[w_2]) \quad ([w_1], [w_2] \in [v]^\alpha).
$$
Since the above hold for every $[v] \in P(H)$, our transformation $\phi$ preserves every angle lying in the interval $\left(0, \max\left(2\alpha,\tfrac{\pi}{2}\right)\right]$.
By Lemma \ref{small-measuredangles_lem} we conlcude that $\phi$ is a Wigner symmetry.

\smallskip

\textbf{Case 2:} when $\dim H \geq 4$ and $\tfrac{\pi}{3} < \alpha < \tfrac{\pi}{2}$.
In this case we have 
$$
\measuredangle\big([\cos\alpha\cdot v + \sin\alpha\cdot u_1],[\cos\alpha\cdot v + \sin\alpha\cdot u_2]\big) = \alpha
$$ 
if and only if 
$$
\sphericalangle(u_1, u_2) \in \left\{ \beta_1, \beta_2 \right\},
$$
where
$$
\beta_1 = \arccos\left(1 - \frac{1}{1+\cos\alpha}\right) \in \left(\arccos\tfrac{1}{3},\tfrac{\pi}{2}\right)
$$ 
and 
$$
\beta_2 = \arccos\left(1 - \frac{1}{1-\cos\alpha}\right)  \in \left(\tfrac{\pi}{2},\pi\right).
$$
Let us observe that, similarly as in the previous case, if one was able to show that every bijection $\psi\colon S_{H\ominus [v]} \to S_{H\ominus [v]}$ which satisfies
\begin{equation}\label{sphere-b1-or-b2-pres_eq}
\sphericalangle(x,y) \in \{\beta_1,\beta_2\} \;\iff\; \sphericalangle(\psi(x),\psi(y)) \in \{\beta_1,\beta_2\}
\end{equation}
is a $\sphericalangle$-isometry, then one could conclude that $\phi$ is a Wigner symmetry. 
In fact, that is what will be done here.

Let $x,y\in S_{H\ominus [v]}$ be two unit vectors with angle $\gamma := \sphericalangle(x,y) \neq 0$.
Obviously, we have
\begin{equation}\label{card}
\begin{gathered}
\card\left((x^{(\beta_1)} \cup x^{(\beta_2)})\cap (y^{(\beta_1)} \cup y^{(\beta_2)})\right) \\
= \card\left(x^{(\beta_1)} \cap y^{(\beta_1)}\right) + 2\cdot\card\left(x^{(\beta_1)} \cap y^{(\beta_2)}\right) + \card\left(x^{(\beta_2)} \cap y^{(\beta_2)}\right).
\end{gathered}
\end{equation}
Let us observe that $\cos\beta_1 < -\cos\beta_2$, which implies $\beta_1+\beta_2 > \pi$.
By Lemma \ref{sphere-0-1-element_lem} we obtain the following equivalences:
\begin{equation}\label{iff1}
\card\left(x^{(\beta_1)} \cap y^{(\beta_1)}\right) = 1 \;\iff\; \gamma = 2\beta_1,
\end{equation}
\begin{equation}\label{iff2}
x^{(\beta_1)} \cap y^{(\beta_1)} = \emptyset \;\iff\; \gamma > 2\beta_1,
\end{equation}
\begin{equation}\label{iff3}
\card\left(x^{(\beta_1)} \cap y^{(\beta_2)}\right) = 1 \;\iff\; \gamma \in \{2\pi-\beta_1-\beta_2, \beta_2-\beta_1\},
\end{equation}
\begin{equation}\label{iff4}
x^{(\beta_1)} \cap y^{(\beta_2)} = \emptyset \;\iff\; \gamma > 2\pi-\beta_1-\beta_2 \text{ or } \gamma < \beta_2-\beta_1,
\end{equation}
\begin{equation}\label{iff5}
\card\left(x^{(\beta_2)} \cap y^{(\beta_2)}\right) = 1 \;\iff\; \gamma = 2\pi-2\beta_2,
\end{equation}
and
\begin{equation}\label{iff6}
x^{(\beta_2)} \cap y^{(\beta_2)} = \emptyset \;\iff\; \gamma > 2\pi-2\beta_2.
\end{equation}
Also, one can easily see that if $\dim (H\ominus [v]) = 3$, then $\card\left(x^{(\beta_j)} \cap y^{(\beta_l)}\right) \in \{0,1,2\}$ is satisfied, and if $\dim (H\ominus [v]) \geq 4$, then we have $\card\left(x^{(\beta_j)} \cap y^{(\beta_l)}\right) \in \{0,1,\infty\}$ $(j,l\in\{1,2\})$. 

Our next step is to compare the following four quantities: $2\beta_1, 2\pi -\beta_1 - \beta_2, \beta_2 - \beta_1, 2\pi -2\beta_2$.
Since we have $\tfrac{\pi}{3} < \beta_1$, we easily get that $3\beta_1 > \beta_2$, and thus $2\beta_1 > \beta_2-\beta_1$.
By similar observations, we obtain the following inequalities: $2\beta_1 > 2\pi - 2\beta_2$, $2\pi -\beta_1 - \beta_2 > \beta_2 - \beta_1$, $2\pi-\beta_1-\beta_2 > 2\pi -2\beta_2$; and the following relations: $2\beta_1, 2\pi -\beta_1 - \beta_2, \beta_2 - \beta_1, 2\pi -2\beta_2 \in (0,\pi)$.

Now, let us solve the equation $2\beta_1 = 2\pi-\beta_1-\beta_2$, which is equivalent to $3\beta_1 = 2\pi - \beta_2$.
Since both sides are in the interval $\left(\pi,\tfrac{3\pi}{2}\right)$, this holds if and only if
$$
1 - \frac{1}{1-\cos\alpha} 
= \cos \beta_2 
= \cos (3\beta_1) 
= 4 \cos^3 (\beta_1) - 3 \cos (\beta_1)
$$
$$
= 4\left(1 - \frac{1}{1+\cos\alpha}\right)^3 - 3 \left(1 - \frac{1}{1+\cos\alpha}\right).
$$
By a quite straightforward calculation we get that this holds exactly when $\cos\alpha \in \{-\tfrac{1}{\sqrt{5}},0,\tfrac{1}{\sqrt{5}}\}$.
Clearly, the first two possibilities cannot occur, therefore we conclude the only solution $\alpha = \arccos\left(\tfrac{1}{\sqrt{5}}\right)$.

Next, we examine the equation $\beta_2-\beta_1 = 2\pi-2\beta_2$, which is the same as $3\beta_2 = 2\pi+\beta_1$.
This implies the following:
$$
1 - \frac{1}{1+\cos\alpha}
= \cos \beta_1 
= \cos (3\beta_2)
= 4 \cos^3 (\beta_2) - 3 \cos (\beta_2)
$$
$$
= 4\left(1 - \frac{1}{1-\cos\alpha} 
\right)^3 - 3 \left(1 - \frac{1}{1-\cos\alpha} 
\right).
$$
Again, by a quite straightforward calculation we obtain that the only solution is $\alpha = \arccos\left(\tfrac{1}{\sqrt{5}}\right)$.

By the above observations we have at most the following six possibilities concerning the order of the four quantities $2\beta_1, 2\pi -\beta_1 - \beta_2, \beta_2 - \beta_1, 2\pi -2\beta_2$: 
\begin{itemize}
\item[1.] $2\beta_1 = 2\pi-\beta_1-\beta_2 > \beta_2-\beta_1 = 2\pi-2\beta_2$;
\item[2.] $2\beta_1 > 2\pi-\beta_1-\beta_2 > \beta_2-\beta_1 > 2\pi-2\beta_2$;
\item[3.] $\dim H = 4$ and $2\pi-\beta_1-\beta_2 > 2\beta_1 > \beta_2-\beta_1 > 2\pi-2\beta_2$;
\item[4.] $\dim H \geq 5$ and $2\pi-\beta_1-\beta_2 > 2\beta_1 > \beta_2-\beta_1 > 2\pi-2\beta_2$;
\item[5.] $2\beta_1 > 2\pi-\beta_1-\beta_2 > 2\pi-2\beta_2 > \beta_2-\beta_1$;
\item[6.] $2\pi-\beta_1-\beta_2 > 2\beta_1 > 2\pi-2\beta_2 > \beta_2-\beta_1$.
\end{itemize}
(In fact, one of these possibilities cannot happen, though we prefer to handle that case as well rather than give some long calculation which would verify its impossibility.)
Using \eqref{iff1}-\eqref{iff6}, in each of the above possibilities we get the following equivalences:
\begin{itemize}
\item[1.] $\card\left((x^{(\beta_1)} \cup x^{(\beta_2)})\cap (y^{(\beta_1)} \cup y^{(\beta_2)})\right) = 3 \;\iff\; \gamma = 2\beta_1 = 2\pi-\beta_1-\beta_2$;
\item[2.] $\card\left((x^{(\beta_1)} \cup x^{(\beta_2)})\cap (y^{(\beta_1)} \cup y^{(\beta_2)})\right) = 1 \;\iff\; \gamma = 2\beta_1$;
\item[3.] $\card\left((x^{(\beta_1)} \cup x^{(\beta_2)})\cap (y^{(\beta_1)} \cup y^{(\beta_2)})\right) = 5 \;\iff\; \gamma = 2\beta_1$;
\item[4.] $\card\left((x^{(\beta_1)} \cup x^{(\beta_2)})\cap (y^{(\beta_1)} \cup y^{(\beta_2)})\right) = 2 \;\iff\; \gamma = 2\pi-\beta_1-\beta_2$;
\item[5.] $\card\left((x^{(\beta_1)} \cup x^{(\beta_2)})\cap (y^{(\beta_1)} \cup y^{(\beta_2)})\right) = 1 \;\iff\; \gamma = 2\beta_1$;
\item[6.] $\card\left((x^{(\beta_1)} \cup x^{(\beta_2)})\cap (y^{(\beta_1)} \cup y^{(\beta_2)})\right) = 2 \;\iff\; \gamma = 2\pi-\beta_1-\beta_2$.
\end{itemize}
Here we only give the verification of the second point above, the others are quite similar.
For this we consider the following table:
\begin{center}
\begin{tabular}{|c|c|c|c|}
\hline
& $\card\left(x^{(\beta_1)} \cap y^{(\beta_1)}\right)$ & $\card\left(x^{(\beta_1)} \cap y^{(\beta_2)}\right)$ & $\card\left(x^{(\beta_2)} \cap y^{(\beta_2)}\right)$ \\
\hline
$0 < \gamma < 2\pi-2\beta_2$ & $2\leq$ & 0 & $2\leq$\\
\hline
$\gamma = 2\pi-2\beta_2$ & $2\leq$ & 0 & 1 \\
\hline
$2\pi-2\beta_2 < \gamma < \beta_2-\beta_1$ & $2\leq$ & 0 & 0\\
\hline
$\gamma = \beta_2-\beta_1$ & $2\leq$ & 1 & 0\\
\hline
$\beta_2-\beta_1 < \gamma < 2\pi-\beta_1-\beta_2$ & $2\leq$ & $2\leq$ & 0\\
\hline
$\gamma = 2\pi-\beta_1-\beta_2$ & $2\leq$ & 1 & 0\\
\hline
$2\pi-\beta_1-\beta_2 < \gamma < 2\beta_1$ & $2\leq$ & 0 & 0\\
\hline
$\gamma = 2\beta_1$ & 1 & 0 & 0\\
\hline
$2\beta_1 < \gamma \leq \pi$ & 0 & 0 & 0\\
\hline
\end{tabular}
\end{center}
Clearly, \eqref{card} gives us the desired conclusion.

Now, in each of the cases above we conclude that $\psi$ preserves either the angle $2\beta_1$ or $2\pi-\beta_1-\beta_2$ in both directions.
The only thing which needs to be shown is that these angles cannot be $\tfrac{\pi}{2}$.
This is obvious for the first one.
Concerning the second one, we know that $\beta_1 < \tfrac{\pi}{2}$ and $\beta_2 < \pi$, thus we easily obtain $2\pi-\beta_1-\beta_2 > \tfrac{\pi}{2}$, which completes this case.

\smallskip

\textbf{Case 3:} when $\dim H \geq 4$ and $\alpha = \tfrac{\pi}{3}$.
Then, similarly as in the previous case, we only need to show that any bijective map $\psi\colon S_{H\ominus [v]} \to S_{H\ominus [v]}$ which satisfies
\begin{equation}
\sphericalangle(x,y) \in \left\{\arccos\left(\tfrac{1}{3}\right),\pi\right\} \;\iff\; \sphericalangle(\psi(x),\psi(y)) \in \left\{\arccos\left(\tfrac{1}{3}\right),\pi\right\}
\end{equation}
is a $\sphericalangle$-isometry.
In order to prove this it is enough to see that $\psi$ preserves the angle $\pi$ in both directions.
But this is straightforward form the following equivalence, where we use the notation $\beta := \arccos\left(\tfrac{1}{3}\right)$:
$$
\left(x^{(\beta)}\cup x^{(\pi)}\right) \cap \left(y^{(\beta)}\cup y^{(\pi)}\right) = \emptyset \text{ and } y\in x^{(\beta)}\cup x^{(\pi)} \;\iff\; y=-x.
$$
This completes the proof of the present case.

\smallskip

\textbf{Case 4:} when $\dim H = 3$ and $0 < \alpha < \tfrac{\pi}{3}$.
Here we have $\dim (H\ominus [v]) = 2$ and therefore, unlike in the previous cases, we cannot apply Theorem \ref{Emain}.
Instead, we shall utilize Lemma \ref{small-measuredangles_lem}.
As in Case 1, one can easily see that the transformation $\psi\colon S_{H\ominus [v]} \to S_{H\ominus [v]}$ preserves the angle $\beta := \arccos \left(\frac{\cos\alpha}{1+\cos\alpha}\right) \in (\tfrac{\pi}{3},\arccos\tfrac{1}{3}) \subset (\tfrac{\pi}{3},\tfrac{\pi}{2})$ in both directions.
Thus $\psi$ also preserves the angle $2\beta \in (\tfrac{2\pi}{3},\pi)$ in both directions, which is a consequence of the following observation:
$$
\sphericalangle(x,y) = 2\beta \;\iff\; \card\left(x^{(\beta)}\cap y^{(\beta)}\right) = 1 \quad (x,y,\in S_{H\ominus [v]}).
$$
Therefore from \eqref{convert-angle-from-proj-tosphere_eq} we conclude that $\phi$ preserves the angle 
\begin{equation}\label{gamma}
\gamma := \arccos\left(\left|\cos^2\alpha+\cos(2\beta)\sin^2\alpha\right|\right) \in \left(0,\tfrac{\pi}{2}\right]
\end{equation}
in both directions.
An elementary calculation gives us the following:
$$
\cos^2\alpha+\cos(2\beta)\sin^2\alpha = \cos^2\alpha+(2\cos^2\beta-1)\sin^2\alpha 
$$
$$
= \cos^2\alpha+\left(2\left(\frac{\cos\alpha}{1+\cos\alpha}\right)^2-1\right)(1-\cos^2\alpha) = 4 \cos\alpha+\frac{4}{\cos\alpha+1}-5.
$$
Therefore \eqref{gamma} becomes 
\begin{equation}\label{gamma2}
\gamma = \arccos \left(\left| 4 \cos\alpha+\frac{4}{\cos\alpha+1}-5 \right|\right).
\end{equation}

Next, an elementary calculation verifies the following inequality:
$$
2c^2 - 1 < 4c + \frac{4}{c+1}-5 < c \qquad (\tfrac{1}{2} \leq c < 1).
$$
It is also easy to show that for every $c \in \left[\tfrac{1}{2}, 1\right)$ we have 
$$
4c + \frac{4}{c+1}-5 < 0 \;\iff\; \tfrac{1}{2} \leq c < \tfrac{1}{8}(1+\sqrt{17}),
$$
and 
$$
4c + \frac{4}{c+1}-5 = 0 \;\iff\; c = \tfrac{1}{8}(1+\sqrt{17}).
$$
Again, by elementary computations, we obtain
$$
\left|4c + \frac{4}{c+1} - 5\right| = 5 - 4c - \frac{4}{c+1} < 1 - 2c^2 \leq c \qquad (\tfrac{1}{2} \leq c < \tfrac{1}{8}(1+\sqrt{17}))
$$
which gives us 
\begin{equation}\label{ceq}
2c^2 - 1 < \left| 4c + \frac{4}{c+1} - 5 \right| < c \qquad (\tfrac{1}{2} \leq c < 1).
\end{equation}
Clearly, \eqref{gamma2} and \eqref{ceq} implies $\gamma \in (\alpha,2\alpha)\cap\left(0,\tfrac{\pi}{2}\right]$.

Now, we consider four different possibilities.
First, if $0 < \alpha < \tfrac{\pi}{4}$, then $0 < \gamma < \tfrac{\pi}{2}$, and by Lemmas \ref{proj-diff-sum_lem} and \ref{multiple-measuredangle_lem} our map $\phi$ preserves the angles $2\alpha-\gamma \in \left(0,\tfrac{\pi}{2}\right)$ and $\gamma-\alpha \in \left(0,\tfrac{\pi}{2}\right)$ in both directions. 
Since one of these angles is less than or equal to $\tfrac{\alpha}{2}$, by a recursion we obtain a sequence of positive angles converging to zero such that all of them is preserved by $\phi$ in both directions.
Therefore by Lemma \ref{small-measuredangles_lem} we conclude that $\phi$ is a Wigner symmetry.

Second, by \eqref{gamma2} we have $\gamma = \tfrac{\pi}{2}$ if and only if $\alpha = \arccos\frac{1+\sqrt{17}}{8}$.
If this happens, then Uhlhorn's theorem implies that $\phi$ is a Wigner symmetry.

Third, we assume that $\tfrac{\pi}{4} \leq \alpha < \arccos\frac{1+\sqrt{17}}{8}$ is satisfied, and we claim that there exists a number $0<q<1$ such that we have $\gamma < (1+q)\alpha$ for every such $\alpha$.
In fact, if this was not the case, then $\left| 4c + \frac{4}{c+1} - 5 \right|$ could be arbitrarily close to $2c^2-1$ on the interval $\left[\frac{1+\sqrt{17}}{8}, \tfrac{1}{\sqrt{2}}\right]$.
But then, continuity of the functions and compactness of the interval would obtain that actually there is a number $c\in \left[\frac{1+\sqrt{17}}{8}, \tfrac{1}{\sqrt{2}}\right]$ for which these two functions coincide, which would contradict to \eqref{ceq}.
Since by Lemma \ref{proj-diff-sum_lem} the map $\phi$ preserves $\gamma-\alpha < q\alpha$, by a recursion we get a sequence of positive angles $\{\alpha_n\}_{n=1}^\infty$ converging to zero such that they are preserved by $\phi$ in both directions.
This implies that $\phi$ has to be a Wigner symmetry.

Finally, let us suppose that we have $\arccos\frac{1+\sqrt{17}}{8} < \alpha < \tfrac{\pi}{3}$.
Since $\tfrac{\pi}{4} \leq \arccos\frac{1+\sqrt{17}}{8}$, we get $\tfrac{\gamma}{\alpha} < \tfrac{\pi/2}{\arccos\frac{1+\sqrt{17}}{8}} < 1.8$ for every such $\alpha$.
But $\phi$ preserves the angle $\gamma - \alpha < 0.8\cdot\alpha$, therefore we have the following two possibilities.
Either by a recursion we obtain a sequence of positive angles converging to zero such that all of them is preserved by $\phi$ in both directions, or after some steps we conclude that $\phi$ preserves the angle $\arccos\frac{1+\sqrt{17}}{8}$, and thus orthogonality in both directions.
Both of them imply that $\phi$ is a Wigner symmetry.

\smallskip

\textbf{Case 5:} when $\dim H = 3$ and $\tfrac{\pi}{3} < \alpha < \tfrac{\pi}{2}$.
We claim that 
\begin{equation}\label{3}
\measuredangle([x],[y]) = \pi-2\alpha \;\iff\; \card([x]^{\alpha} \cap [y]^{\alpha}) = 3.
\end{equation}
In order to see this let $[x]$ and $[y]$ be two different lines and set $\gamma := \measuredangle([x],[y]) \in \left(0,\tfrac{\pi}{2}\right]$.
There exists two orthogonal unit vectors $e_1, e_2 \in H$ such that we have 
\begin{equation}\label{xyline}
[x] = [\cos\tfrac{\gamma}{2}\cdot e_1 + \sin\tfrac{\gamma}{2}\cdot e_2], \;\; [y] = [\cos\tfrac{\gamma}{2}\cdot e_1 - \sin\tfrac{\gamma}{2}\cdot e_2].
\end{equation}
We also fix a unit vector $e_3$ which is orthogonal to $e_1$ and $e_2$.
For $t_1, t_2, t_3 \in \R$, $t_1^2+t_2^2+t_3^2 = 1$ we have 
$$
[t_1\cdot e_1 + t_2\cdot e_2 + t_3\cdot e_3] \in [x]^{\alpha} \cap [y]^{\alpha}
$$ 
if and only if
$$
\cos\alpha = \left|t_1 \cos\tfrac{\gamma}{2} + t_2 \sin\tfrac{\gamma}{2}\right| = \left|t_1 \cos\tfrac{\gamma}{2} - t_2 \sin\tfrac{\gamma}{2}\right|.
$$
We immediately conclude that either $t_1 = 0$, or $t_2 = 0$.
In the first possibility we compute $|t_2| = \frac{\cos\alpha}{\sin\tfrac{\gamma}{2}} > 0$, moreover, we have
$$
\frac{\cos\alpha}{\sin\tfrac{\gamma}{2}} < 1 \;\iff\; \gamma > \pi-2\alpha
$$
and
$$
\frac{\cos\alpha}{\sin\tfrac{\gamma}{2}} = 1 \;\iff\; \gamma = \pi-2\alpha.
$$
In the second possibility we have $|t_1| = \frac{\cos\alpha}{\cos\tfrac{\gamma}{2}} \in (0,1)$.
Therefore if $\gamma < \pi-2\alpha$, then $\card\left([x]^{\alpha} \cap [y]^{\alpha}\right) \leq 2$.
Otherwise, we have
\begin{equation}\label{xametszetya}
\begin{gathered}
{[x]}^\alpha \cap [y]^{\alpha} = \Bigg\{ \left[\frac{\cos\alpha}{\sin\tfrac{\gamma}{2}}\cdot e_2 + \sqrt{1-\frac{\cos^2\alpha}{\sin^2\tfrac{\gamma}{2}}}\cdot e_3\right], \left[\frac{\cos\alpha}{\sin\tfrac{\gamma}{2}}\cdot e_2 - \sqrt{1-\frac{\cos^2\alpha}{\sin^2\tfrac{\gamma}{2}}}\cdot e_3\right], \\
\left[\frac{\cos\alpha}{\cos\tfrac{\gamma}{2}} \cdot e_1 + \sqrt{1-\frac{\cos^2\alpha}{\cos^2\tfrac{\gamma}{2}}}\cdot e_3\right], \left[\frac{\cos\alpha}{\cos\tfrac{\gamma}{2}} \cdot e_1 - \sqrt{1-\frac{\cos^2\alpha}{\cos^2\tfrac{\gamma}{2}}}\cdot e_3\right] \Bigg\}.
\end{gathered}
\end{equation}
Clearly, the third and fourth elements above are always different.
Furthermore, they are also different from the first and second elements.
Therefore we have $\card\left([x]^{\alpha} \cap [y]^{\alpha}\right) = 3$ if and only if the first and second elements coincide, which yields \eqref{3}.

Now, since we have
$$
\measuredangle([x],[y]) = \pi-2\alpha
\;\iff\; \card([x]^{\alpha} \cap [y]^{\alpha}) = 3
\;\iff\; \card(\phi([x])^{\alpha} \cap \phi([y])^{\alpha}) = 3 
$$
$$
\;\iff\; \measuredangle(\phi([x]),\phi([y])) = \pi-2\alpha,
$$
an application of Case 4 ensures that $\phi$ is indeed a Wigner symmetry.

\smallskip

\textbf{Case 6:} when $\dim H = 3$ and $\alpha = \tfrac{\pi}{3}$.
We claim that we have $\measuredangle([x],[y]) = \tfrac{\pi}{2}$ if and only if $[x]^{\pi/3} \cap [y]^{\pi/3} = \{[u_1],[u_2],[u_3],[u_4]\}$ is a four-element set with $\measuredangle([u_1],[u_3]) = \measuredangle([u_1],[u_4]) = \measuredangle([u_2],[u_3]) = \measuredangle([u_2],[u_4]) = \tfrac{\pi}{3}$.
Using the notations of the previous case, we have $\card\left([x]^{\pi/3} \cap [y]^{\pi/3}\right) \leq 3$ if and only if $\gamma \leq \tfrac{\pi}{3}$, otherwise $\card\left([x]^{\pi/3} \cap [y]^{\pi/3}\right) = 4$ and $[x]^{\pi/3} \cap [y]^{\pi/3} = \{[v_1(\gamma)], [v_2(\gamma)], [v_3(\gamma)], [v_4(\gamma)]\}$ where we use the following notations:
$$
v_1(\gamma) = \frac{1/2}{\sin\tfrac{\gamma}{2}}\cdot e_2 + \sqrt{1-\frac{1/4}{\sin^2\tfrac{\gamma}{2}}}\cdot e_3, \quad 
v_2(\gamma) = \frac{1/2}{\sin\tfrac{\gamma}{2}}\cdot e_2 - \sqrt{1-\frac{1/4}{\sin^2\tfrac{\gamma}{2}}}\cdot e_3,
$$
$$
v_3(\gamma) = \frac{1/2}{\cos\tfrac{\gamma}{2}} \cdot e_1 + \sqrt{1-\frac{1/4}{\cos^2\tfrac{\gamma}{2}}}\cdot e_3, \quad
v_4(\gamma) = \frac{1/2}{\cos\tfrac{\gamma}{2}} \cdot e_1 - \sqrt{1-\frac{1/4}{\cos^2\tfrac{\gamma}{2}}}\cdot e_3.
$$
A straightforward computation gives the following:
$$
\langle v_1(\gamma), v_3(\gamma)\rangle = - \langle v_1(\gamma), v_4(\gamma)\rangle = - \langle v_2(\gamma), v_3(\gamma)\rangle = \langle v_2(\gamma), v_4(\gamma)\rangle
$$
$$
= \sqrt{1 - \frac{1/4}{\sin^2\tfrac{\gamma}{2}}} \cdot \sqrt{1 - \frac{1/4}{\cos^2\tfrac{\gamma}{2}}},
$$
where $\sin^2\tfrac{\gamma}{2}$ runs through the interval $\left(\tfrac{1}{4},\tfrac{1}{2}\right]$ if $\gamma$ runs through $\left(\tfrac{\pi}{3},\tfrac{\pi}{2}\right]$.
Denoting $\sin^2\tfrac{\gamma}{2}$ by $u$ and observing that the equation $\left(1-\tfrac{1}{4u}\right)\left(1-\tfrac{1}{4(1-u)}\right) = \tfrac{1}{4}$ has the only solution $u = \tfrac{1}{2}$ verifies our claim.

Finally, similarly as in the last paragraph of the previous case, one can show easily that $\phi$ preserves orthogonality in both directions.
Therefore $\phi$ has to be a Wigner symmetry because of Uhlhorn's theorem.
\end{proof}

Our final aim is to prove our theorem for complex inner product spaces of dimension at least three.
But before doing so, we provide the following crucial lemma.
We will use the notation
$$
\diam_\measuredangle\left([v]^\alpha \cap [w]^\alpha\right) = \sup\left(\measuredangle([u_1],[u_2]) \colon [u_1], [u_2] \in [v]^\alpha \cap [w]^\alpha\right).
$$

\begin{lem}\label{proj-crucial_lem}
Let $\alpha\in\left(0,\tfrac{\pi}{4}\right)$, $H$ be a complex inner product space with $\dim H \geq 3$, and $[v],[w] \in P(H)$ be two different lines with $\measuredangle([v],[w]) =: \gamma \in (0,2\alpha)$.
Then 
\begin{equation}\label{szogdiam}
\diam_\measuredangle\left([v]^\alpha \cap [w]^\alpha\right) 
= 2 \cdot \arccos \sqrt{ \frac{\cos^2\alpha-\sin^2\left(\tfrac{\gamma}{2}\right)}{\cos^2\left(\tfrac{\gamma}{2}\right)-\sin^2\left(\tfrac{\gamma}{2}\right)} }.
\end{equation}

Moreover, there exist exactly two lines $[u_1], [u_2] \in [v]^\alpha \cap [w]^\alpha$ with $\measuredangle([u_1],[u_2]) = \diam_\measuredangle\left([v]^\alpha \cap [w]^\alpha\right)$; 
and if $0 < \beta < \diam_\measuredangle\left([v]^\alpha \cap [w]^\alpha\right)$, then there exist at least three different lines $[u_0], [u_1], [u_2] \in [v]^\alpha \cap [w]^\alpha$ with $\measuredangle([u_0],[u_1]) = \measuredangle([u_0],[u_2]) = \beta$.
\end{lem}

\begin{proof}
There is an orthonormal system $\{e_1, e_2\}$ such that we have 
\begin{equation}\label{vw}
[v] = \left[\cos\left(\tfrac{\gamma}{2}\right)\cdot e_1 + \sin\left(\tfrac{\gamma}{2}\right)\cdot e_2\right] \;\; \text{and} \;\; [w] = \left[\cos\left(\tfrac{\gamma}{2}\right)\cdot e_1 - \sin\left(\tfrac{\gamma}{2}\right)\cdot e_2\right].
\end{equation}
For an arbitrary $[u]\in P(H)$, we can find a vector $e_3\in H$, $\|e_3\| = 1$, $e_3\perp e_j$ $(j=1,2)$ and numbers $\lambda \in \irC$, $\delta,\varepsilon \in [0,\tfrac{\pi}{2}]$ such that 
$$
[u] = [\cos\delta \cdot e_1 + \lambda\sin\delta\cos\varepsilon\cdot e_2 + \sin\delta\sin\varepsilon\cdot e_3].
$$
This line lies in $[v]^\alpha \cap [w]^\alpha$ if and only if we have
$$
\cos\alpha
=
\left|\cos\left(\tfrac{\gamma}{2}\right)\cos\delta + \lambda\sin\left(\tfrac{\gamma}{2}\right)\sin\delta\cos\varepsilon\right|
=
\left|\cos\left(\tfrac{\gamma}{2}\right)\cos\delta - \lambda\sin\left(\tfrac{\gamma}{2}\right)\sin\delta\cos\varepsilon\right|.
$$
The second equation implies that we have either $\delta = 0$, or $\delta = \tfrac{\pi}{2}$, or $\varepsilon = \tfrac{\pi}{2}$ or $\lambda \in \{-i,i\}$.
For the first one, we get $\gamma = 2\alpha$ which is impossible under our assumptions.
Concerning the second one, the contradiction $\tfrac{1}{\sqrt{2}} < \cos\alpha = \left|\sin\left(\tfrac{\gamma}{2}\right)\cos\varepsilon\right| < \tfrac{1}{\sqrt{2}}$ is yielded.
The third case implies $\cos\delta = \frac{\cos\alpha}{\cos\left(\tfrac{\gamma}{2}\right)} \in (0,1)$.
From the last possibility we obtain the following:
$$
\cos^2\alpha
= 
\cos^2\left(\tfrac{\gamma}{2}\right)\cos^2\delta + \sin^2\left(\tfrac{\gamma}{2}\right)\sin^2\delta\cos^2\varepsilon
$$
$$
=
\cos^2\left(\tfrac{\gamma}{2}\right)\cos^2\delta + \sin^2\left(\tfrac{\gamma}{2}\right)(1-\cos^2\delta)\cos^2\varepsilon,
$$
from which we conclude
$$
\cos\delta = \sqrt{ \frac{\cos^2\alpha-\sin^2\left(\tfrac{\gamma}{2}\right)\cos^2\varepsilon}{\cos^2\left(\tfrac{\gamma}{2}\right)-\sin^2\left(\tfrac{\gamma}{2}\right)\cos^2\varepsilon} }.
$$
We note that the above fraction is positive, since we have $\cos\left(\tfrac{\gamma}{2}\right) > \cos\alpha > \tfrac{1}{\sqrt{2}} > \sin\left(\tfrac{\gamma}{2}\right)$.
We also point out that the $\varepsilon = \tfrac{\pi}{2}$ case in the present possibility covers the before mentioned third possibility.
Therefore we conclude
\begin{equation}\label{v-cap-w_eq}
\begin{gathered}[]
[v]^\alpha \cap [w]^\alpha 
= 
\Bigg\{ \Bigg[ \sqrt{ \frac{\cos^2\alpha-\sin^2\left(\tfrac{\gamma}{2}\right)\cos^2\varepsilon}{\cos^2\left(\tfrac{\gamma}{2}\right)-\sin^2\left(\tfrac{\gamma}{2}\right)\cos^2\varepsilon} } \cdot e_1 + \\
+ \lambda \sqrt{ \frac{\cos^2\left(\tfrac{\gamma}{2}\right) - \cos^2\alpha}{\cos^2\left(\tfrac{\gamma}{2}\right)-\sin^2\left(\tfrac{\gamma}{2}\right)\cos^2\varepsilon} } \cos\varepsilon\cdot e_2
+ \sqrt{ \frac{\cos^2\left(\tfrac{\gamma}{2}\right) - \cos^2\alpha}{\cos^2\left(\tfrac{\gamma}{2}\right)-\sin^2\left(\tfrac{\gamma}{2}\right)\cos^2\varepsilon} }\sin\varepsilon\cdot e_3 \Bigg] \colon \\
\colon \|e_3\| = 1, \, e_3\perp e_1, \, e_3 \perp e_2,\, \varepsilon\in\left[0,\tfrac{\pi}{2}\right], \, \lambda\in\{-i,i\}\Bigg\}.
\end{gathered}
\end{equation}

Since the function $g_{x,y}(t) = \frac{x-t}{y-t}$ is strictly decreasing on $[0,x]$ if $0<x<y$, and strictly increasing on $[0,y]$ if $0<y<x$, we estimate as follows:
\begin{equation}\label{angle-est_eq}
\begin{gathered}
\sup\Big\{\measuredangle([u],[e_1]) \colon [u] \in [v]^\alpha \cap [w]^\alpha\Big\}
 \\
= \arccos\left( \inf\left\{\sqrt{ \frac{\cos^2\alpha-\sin^2\left(\tfrac{\gamma}{2}\right)\cos^2\varepsilon}{\cos^2\left(\tfrac{\gamma}{2}\right)-\sin^2\left(\tfrac{\gamma}{2}\right)\cos^2\varepsilon} } \colon \varepsilon\in\left[0,\tfrac{\pi}{2}\right] \right\} \right)
 \\
= \arccos\left( \min\left\{\sqrt{ \frac{\cos^2\alpha - t}{\cos^2\left(\tfrac{\gamma}{2}\right) - t} } \colon t\in\left[0,\sin^2\left(\tfrac{\gamma}{2}\right)\right] \right\} \right) 
 \\
= \arccos \sqrt{ \frac{\cos^2\alpha-\sin^2\left(\tfrac{\gamma}{2}\right)}{\cos^2\left(\tfrac{\gamma}{2}\right)-\sin^2\left(\tfrac{\gamma}{2}\right)} } 
= \arccos \sqrt{ \tfrac{1}{2} \cdot \frac{\cos^2\alpha-\sin^2\left(\tfrac{\gamma}{2}\right)}{\tfrac{1}{2}-\sin^2\left(\tfrac{\gamma}{2}\right)} }
 \\
\leq \arccos \sqrt{ \tfrac{1}{2} \cdot \frac{\cos^2\alpha-0}{\tfrac{1}{2}-0} } 
= \alpha < \tfrac{\pi}{4}. 
\end{gathered}
\end{equation}
The left-hand side of \eqref{angle-est_eq} is strictly decreasing in $\gamma$.
Moreover, the supremum is attained exactly when $\varepsilon = 0$, in particular, if and only if we have either
\begin{equation}\label{max-u1_eq}
[u] = \left[ \sqrt{ \frac{\cos^2\alpha-\sin^2\left(\tfrac{\gamma}{2}\right)}{\cos^2\left(\tfrac{\gamma}{2}\right)-\sin^2\left(\tfrac{\gamma}{2}\right)} } \cdot e_1 + i \cdot \sqrt{ \frac{\cos^2\left(\tfrac{\gamma}{2}\right) - \cos^2\alpha}{\cos^2\left(\tfrac{\gamma}{2}\right)-\sin^2\left(\tfrac{\gamma}{2}\right)} }\cdot e_2 \right] \; \in \; \left[e_1,e_2\right],
\end{equation}
or
\begin{equation}\label{max-u2_eq}
[u] = \left[ \sqrt{ \frac{\cos^2\alpha-\sin^2\left(\tfrac{\gamma}{2}\right)}{\cos^2\left(\tfrac{\gamma}{2}\right)-\sin^2\left(\tfrac{\gamma}{2}\right)} } \cdot e_1 - i \cdot \sqrt{ \frac{\cos^2\left(\tfrac{\gamma}{2}\right) - \cos^2\alpha}{\cos^2\left(\tfrac{\gamma}{2}\right)-\sin^2\left(\tfrac{\gamma}{2}\right)} }\cdot e_2 \right] \; \in \; \left[e_1,e_2\right].
\end{equation}

Now, we consider the following bijective linear isometry:
$$
R\colon H\to H, \; R(\lambda\cdot e_1 + h) = \lambda\cdot e_1 - h \quad (\lambda\in\C, \, h\in H, \; h\perp e_1),
$$
which clearly satisfies the equations $R([v]) = [w]$, $R([w]) = [v]$ and $R([v]^\alpha \cap [w]^\alpha) = [v]^\alpha \cap [w]^\alpha$.
Let $[u_1],[u_2] \in [v]^\alpha \cap [w]^\alpha$, then, on one hand, we have $\measuredangle([u_1], [u_2]) \leq \measuredangle([u_1], [e_1]) + \measuredangle([e_1], [u_2])$. 
On the other hand, if $[\lambda \cdot e_1 + h] \in [v]^\alpha \cap [w]^\alpha$ $(h\perp e_1$, $|\lambda|^2+\|h\|^2 = 1)$, then by \eqref{angle-est_eq} we have $|\lambda| > \frac{1}{\sqrt{2}}$.
Thus we calculate
$$
\measuredangle([\lambda \cdot e_1 + h], [R(\lambda \cdot e_1 + h)]) = \measuredangle([\lambda \cdot e_1 + h], [\lambda \cdot e_1 - h]) = \arccos\left| |\lambda|^2 - \|h\|^2 \right| 
$$
$$
= \arccos\left| 2\cdot|\lambda|^2 - 1 \right| = 
2 \cdot \arccos|\lambda|= 
2 \cdot \measuredangle([\lambda \cdot e_1 + h], [e_1])
$$
This and \eqref{angle-est_eq} implies that
$$
\diam_\measuredangle\left([v]^\alpha \cap [w]^\alpha\right) = 2 \cdot \sup\left(\measuredangle([u],[e_1]) \colon [u] \in [v]^\alpha \cap [w]^\alpha\right) 
$$
$$
= 2 \cdot \arccos \sqrt{ \frac{\cos^2\alpha-\sin^2\left(\tfrac{\gamma}{2}\right)}{\cos^2\left(\tfrac{\gamma}{2}\right)-\sin^2\left(\tfrac{\gamma}{2}\right)} },
$$
which verifies \eqref{szogdiam}. 
Moreover, the above observations together imply that $\measuredangle([u_1],[u_2]) = \diam_\measuredangle\left([v]^\alpha \cap [w]^\alpha\right)$ and $[u_1],[u_2] \in [v]^\alpha \cap [w]^\alpha$ hold if and only if the lines $[u_1]$ and $[u_2]$ are exactly those defined in \eqref{max-u1_eq} and \eqref{max-u2_eq}.

Finally, let $0 < \beta < \diam_\measuredangle\left([v]^\alpha \cap [w]^\alpha\right)$.
Let $[u]$ be the line defined in \eqref{max-u1_eq}, and $[\widetilde{u}]$ be another line from $[v]^\alpha \cap [w]^\alpha$ such that $\measuredangle([u],[\widetilde{u}]) = \beta$.
The existence of such a line $[\widetilde{u}]$ is straightforward from the pathwise connectedness of $[v]^\alpha \cap [w]^\alpha$, which can be verified from \eqref{v-cap-w_eq}.
Since $\card\left( [v]^\alpha \cap [w]^\alpha \cap [e_1,e_2] \right) = 2$, we obtain that $[\widetilde u] \notin [e_1,e_2]$.
Let $\widetilde{u} = \hat{u} + \check{u}$ such that $\hat u \in [e_1,e_2]$ and $\check u \perp e_j$ $(j=1,2)$.
Clearly, by \eqref{v-cap-w_eq} the line $[\hat{u} - \check{u}]$ lies in $[v]^\alpha \cap [w]^\alpha\setminus \{[\widetilde{u}]\}$, moreover, we have $\measuredangle([\hat{u} - \check{u}],[u]) = \measuredangle([\widetilde{u}],[u]) = \beta$.
This completes our proof.
\end{proof}

We point out that if $\alpha = \tfrac{\pi}{4}$, then we have 
$$
\diam_\measuredangle\left([v]^\alpha \cap [w]^\alpha\right) = \diam_\measuredangle\left([v]^\alpha \cap [w]^\alpha \cap P_{[v],[w]} \right) = \tfrac{\pi}{2}
$$ 
for any two different lines $[v], [w]\in P(H)$, which can be verified by utilizing Bloch's representation.
Therefore \eqref{szogdiam} does not hold for general angles $\alpha$.

We conclude the paper by proving our result on complex projective spaces when the dimension of $H$ is at least three.

\begin{proof}[Proof of Theorem \ref{Umain-complex}]
\textbf{Case 1:} when $0 < \alpha < \tfrac{\pi}{4}$. 
Let $[v]$ and $[w]$ be two different lines.
By Lemma \ref{proj-0-1-element_lem} we have $\card\left([v]^\alpha \cap [w]^\alpha\right) \leq 1$ if and only if $\measuredangle([v],[w]) \geq 2\alpha$.
Let us examine the following equation:
$$
\alpha = 2 \cdot \arccos \sqrt{ \frac{\cos^2\alpha-\sin^2\left(\tfrac{\gamma}{2}\right)}{\cos^2\left(\tfrac{\gamma}{2}\right)-\sin^2\left(\tfrac{\gamma}{2}\right)} }.
$$
It is clear that as $\gamma \to 0+$, the right-hand side tends to $2\alpha$, and as $\gamma \to 2\alpha-$, the limit of the right-hand side is zero.
It was noted in the proof of Lemma \ref{proj-crucial_lem} that the right-hand side is strictly decreasing.
Therefore the above equation has a unique solution $\gamma_0 \in (0,2\alpha)$.
Now, let us observe the following:
$$
\alpha < 2 \cdot \arccos \sqrt{ \frac{\cos^2\alpha-\sin^2\left(\tfrac{\alpha}{2}\right)}{\cos^2\left(\tfrac{\alpha}{2}\right)-\sin^2\left(\tfrac{\alpha}{2}\right)} } = 2 \cdot \arccos \sqrt{ \frac{\cos^2\alpha-\sin^2\left(\tfrac{\alpha}{2}\right)}{\cos\alpha} }
$$
$$
\;\iff\;
\cos^2\alpha - \sin^2\left(\tfrac{\alpha}{2}\right) < \cos\alpha \cos^2\left(\tfrac{\alpha}{2}\right)
\;\iff\;
0 < \tfrac{1}{2}\sin^2\alpha.
$$
Since the latter inequality is valid, we actually have $\gamma_0 \in (\alpha,2\alpha)$.

Next, by Lemma \ref{proj-crucial_lem} we infer
$$
\measuredangle([v],[w]) = \gamma_0 
\;\iff\; \diam_\measuredangle\left([v]^\alpha \cap [w]^\alpha\right) = \alpha
$$
$$
\;\iff\; \card\left( [u]^\alpha\cap[v]^\alpha\cap[w]^\alpha \right) \leq 1 \; (\forall\; [u] \in [v]^\alpha\cap[w]^\alpha)
$$
$$
\;\iff\; \card\left( [x]^\alpha\cap\phi([v])^\alpha\cap\phi([w])^\alpha \right) \leq 1 \; (\forall\; [x] \in \phi([v])^\alpha\cap\phi([w])^\alpha)
$$
$$
\;\iff\; \diam_\measuredangle\left(\phi([v])^\alpha \cap \phi([w])^\alpha\right) = \alpha
\;\iff\; \measuredangle(\phi([v]),\phi([w])) = \gamma_0,
$$
i.e. $\phi$ preserves the angle $\gamma_0$ in both directions.
Lemma \ref{multiple-measuredangle_lem} implies that the angle $2\alpha$ is also preserved in both directions.
Therefore, by Lemma \ref{proj-diff-sum_lem}, our transformation preserves the angles $\gamma_0 - \alpha$ and $2\alpha - \gamma_0$ in both directions.
Since none of these angles are zero and at least one of them is less than or equal to $\tfrac{\alpha}{2}$, a straightforward induction and Lemma \ref{small-measuredangles_lem} complete the proof of this case.

\smallskip

\textbf{Case 2:} when $\alpha = \tfrac{\pi}{4}$. 
Again, let $[v]$ and $[w]$ be two different lines.
If they are orthogonal, then by Lemma \ref{proj-0-1-element_lem} we have
$$
[v]^{\pi/4}\cap[w]^{\pi/4} = \left\{\left[\sqrt{\tfrac{1}{2}} \cdot v + \lambda \sqrt{\tfrac{1}{2}} \cdot w\right] \colon \lambda\in\irC \right\}.
$$
For arbitrary $\lambda,\mu\in\irC$ we compute the following:
$$
\measuredangle \left( \left[ \sqrt{\tfrac{1}{2}} \cdot v + \lambda \sqrt{\tfrac{1}{2}} \cdot w\right], \left[\sqrt{\tfrac{1}{2}} \cdot v + \mu \sqrt{\tfrac{1}{2}} \cdot w \right] \right) = \arccos\tfrac{|1 + \lambda\overline{\mu}|}{2}.
$$
Therefore we conclude that $\card\left([u]^{\pi/4}\cap[v]^{\pi/4}\cap[w]^{\pi/4}\right) = 2$ is satisfied for every $[u] \in [v]^{\pi/4}\cap[w]^{\pi/4}$.

Now, let us assume that $\gamma := \measuredangle([v], [w]) \in \left(0,\tfrac{\pi}{2}\right)$.
We will use the notations of \eqref{vw}.
A direct calculation gives that the following set is a subset of $[v]^{\pi/4} \cap [w]^{\pi/4}$:
\begin{equation}\label{subset}
\begin{gathered}[]
M_{[v],[w]} := \Bigg\{ \Bigg[ \sqrt{ \frac{\tfrac{1}{2}-\sin^2\left(\tfrac{\gamma}{2}\right)\cos^2\varepsilon}{\cos^2\left(\tfrac{\gamma}{2}\right)-\sin^2\left(\tfrac{\gamma}{2}\right)\cos^2\varepsilon} } \cdot e_1 + \\
+ \lambda \sqrt{ \frac{\cos^2\left(\tfrac{\gamma}{2}\right) - \tfrac{1}{2}}{\cos^2\left(\tfrac{\gamma}{2}\right)-\sin^2\left(\tfrac{\gamma}{2}\right)\cos^2\varepsilon} } \cos\varepsilon\cdot e_2 
+ \sqrt{ \frac{\cos^2\left(\tfrac{\gamma}{2}\right) - \tfrac{1}{2}}{\cos^2\left(\tfrac{\gamma}{2}\right)-\sin^2\left(\tfrac{\gamma}{2}\right)\cos^2\varepsilon} }\sin\varepsilon\cdot e_3 \Bigg] \\
\colon \|e_3\| = 1, \, e_3\perp e_1, \, e_3 \perp e_2, \varepsilon\in\left[0,\tfrac{\pi}{2}\right] \, \lambda\in\{-i,i\}\Bigg\}.
\end{gathered}
\end{equation}
In particular, we have
$$
[u] := \left[ \sqrt{ \frac{\tfrac{1}{2}-\sin^2\left(\tfrac{\gamma}{2}\right)}{\cos^2\left(\tfrac{\gamma}{2}\right)-\sin^2\left(\tfrac{\gamma}{2}\right)} } \cdot e_1 + i \sqrt{ \frac{\cos^2\left(\tfrac{\gamma}{2}\right) - \tfrac{1}{2}}{\cos^2\left(\tfrac{\gamma}{2}\right)-\sin^2\left(\tfrac{\gamma}{2}\right)} } \cdot e_2 \right] \; \in \; M_{[v],[w]}
$$
and
$$
[u'] := \left[ \sqrt{ \frac{\tfrac{1}{2}-\sin^2\left(\tfrac{\gamma}{2}\right)}{\cos^2\left(\tfrac{\gamma}{2}\right)-\sin^2\left(\tfrac{\gamma}{2}\right)} } \cdot e_1 - i \sqrt{ \frac{\cos^2\left(\tfrac{\gamma}{2}\right) - \tfrac{1}{2}}{\cos^2\left(\tfrac{\gamma}{2}\right)-\sin^2\left(\tfrac{\gamma}{2}\right)} } \cdot e_2 \right] \; \in \; M_{[v],[w]}.
$$
A simple computation shows that $[u] \perp [u']$.
Since $M_{[v],[w]}$ is pathwise connected, we have a line $[\widetilde{u}] \in M_{[v],[w]}$ with $\measuredangle([\widetilde{u}],[u]) = \tfrac{\pi}{4}$.
As in the last paragraph of the proof of Lemma \ref{proj-crucial_lem}, we can see that $[\widetilde{u}] \notin P_{[e_1],[e_2]}$, and thus we can easily construct infinitely many other lines which are in $[v]^{\pi/4}\cap[w]^{\pi/4}\cap[u]^{\pi/4}$.
Therefore if two different lines $[v]$ and $[w]$ are not orthogonal, then there exists a line $[u] \in [v]^{\pi/4}\cap[w]^{\pi/4}$ such that $\card\left([u]^{\pi/4}\cap[v]^{\pi/4}\cap[w]^{\pi/4}\right) = \infty$.

By the above observations we have the following equivalence-chain:
$$
[v]\perp [w] \;\iff\; \card\left([u]^{\pi/4}\cap[v]^{\pi/4}\cap[w]^{\pi/4}\right) = 2 \;\; (\forall\; [u] \in [v]^{\pi/4}\cap[w]^{\pi/4})
$$
$$
\iff\; \card\left([x]^{\pi/4}\cap\phi([v])^{\pi/4}\cap\phi([w])^{\pi/4}\right) = 2 \;\; (\forall\; [x] \in \phi([v])^{\pi/4}\cap\phi([w])^{\pi/4})
$$
$$
\iff\; \phi([v])\perp \phi([w]).
$$
Therefore, by Uhlhorn's theorem, $\phi$ is a Wigner symmetry which complpetes the proof.
\end{proof}


\section*{Acknowledgment}
The author was also supported by the "Lend\" ulet" Program (LP2012-46/2012) of the Hungarian Academy of Sciences and by the Hungarian National Foundation for Scientific Research (OTKA), Grant No. K115383.


\end{document}